\documentclass[11pt]{article}

\usepackage[paper=letterpaper, margin=1in]{geometry}

\usepackage{ifthen}

\ifthenelse{\isundefined{\GenerateShortVersion}}
{
\def\LongVersion{}
\def\LongVersionEnd{}
\long\def\ShortVersion#1\ShortVersionEnd{}
}
{
\def\ShortVersion{}
\def\ShortVersionEnd{}
\long\def\LongVersion#1\LongVersionEnd{}
}

\newcommand{\Ignore}[1]{\ignorespaces}

\usepackage{amsmath, amssymb}

\usepackage{amsthm}

\usepackage{ifpdf}

\usepackage{authblk}

\ifpdf
\usepackage[pdftex]{graphicx}
\else
\usepackage[dvips]{graphicx}
\fi

\usepackage[%
ocgcolorlinks,%
linkcolor=blue,%
filecolor=blue,%
citecolor=blue,%
urlcolor=blue]{hyperref}

\LongVersion 
\LongVersionEnd 

\ShortVersion 
\usepackage{times}
\ShortVersionEnd 

\ShortVersion 
\renewcommand{\paragraph}[1]{\par\noindent\textbf{#1}}
\ShortVersionEnd 

\newtheorem{theorem}{Theorem}[section]
\newtheorem{lemma}[theorem]{Lemma}
\newtheorem{observation}[theorem]{Observation}
\newtheorem{corollary}[theorem]{Corollary}

\theoremstyle{definition}

\theoremstyle{plain}

\newcommand{\Integers}{\mathbb{Z}}
\newcommand{\Neighbors}{\mathit{N}}
\newcommand{\V}{\mathcal{V}}
\newcommand{\E}{\mathcal{E}}

\renewcommand{\Pr}{\mathbb{P}}

\newcommand{\varParent}{\ensuremath{\mathtt{prnt}}}
\newcommand{\varChildren}{\ensuremath{\mathtt{chld}}}
\newcommand{\varToken}{\ensuremath{\mathtt{tkn}}}
\newcommand{\varTokenDirection}{\ensuremath{\mathtt{tkn\_d}}}
\newcommand{\varTimer}{\ensuremath{\mathtt{tmr}}}
\newcommand{\varTreeNeighbors}{\ensuremath{\mathtt{T\_nbrs}}}
\newcommand{\varOutProposal}{\ensuremath{\mathtt{out\_prop}}}
\newcommand{\varInProposals}{\ensuremath{\mathtt{in\_prop}}}
\newcommand{\varAcceptingProposals}{\ensuremath{\mathtt{accpt}}}
\newcommand{\varRecentPass}{\ensuremath{\mathtt{recent}}}
\newcommand{\varDummy}{\ensuremath{\mathtt{dummy}}}

\newcommand{\ModuleMain}{\ensuremath{\mathtt{Main}}}
\newcommand{\ModuleAccepting}{\ensuremath{\mathtt{Accept}}}
\newcommand{\ModuleProposing}{\ensuremath{\mathtt{Propose}}}
\newcommand{\ModuleTraversal}{\ensuremath{\mathtt{Traverse}}}
\newcommand{\ModuleFindCrossingEdge}{\ensuremath{\mathtt{Cross}}}
\newcommand{\ModuleRootTransfer}{\ensuremath{\mathtt{Trnsfer}}}
\newcommand{\ProcRestart}{\ensuremath{\mathtt{Restart}}}

\newcommand{\msgPassToken}{\ensuremath{\mathtt{pass\_tkn}}}
\newcommand{\msgPropose}{\ensuremath{\mathtt{propose}}}
\newcommand{\msgAccept}{\ensuremath{\mathtt{accept}}}
\newcommand{\msgRootTransfer}{\ensuremath{\mathtt{root\_trns}}}

\newcommand{\Ctr}{\mathit{C}_{\mathrm{tr}}}
\newcommand{\Cph}{\mathit{C}_{\mathrm{ph}}}
\newcommand{\MaxTimer}{\varTimer_{\max}}
\newcommand{\MinTimer}{\varTimer_{\min}}
\newcommand{\Lph}{L_{\mathrm{ph}}}

\newcommand{\Thm}{Thm.}
\newcommand{\Lem}{Lem.}
\newcommand{\Cor}{Cor.}
\newcommand{\Obs}{Obs.}
\newcommand{\Sect}{Sec.}

\title{Communication Efficient Self-Stabilizing Leader Election \\
(Full Version)\footnote{%
The work of X.~D\'efago, T.~Masuzawa, and Y.~Tamura was supported by JST
SICORP Grant Number JPMJSC1606.
The work of Y.~Emek and S.~Kutten was supported by the Israeli Ministry of
Science and Technology (MOST) grant number 3-13565 and by the Technion Hiroshi
Fujiwara Cyber Security Research Center and the Israel National Cyber
Directorate.
The work of T.~Masuzawa was supported by JSPS KAKENHI Grant Number 19H04085.
}}

\author{Xavier D\'efago}
\affil{Tokyo Institute of Technology, Japan.
Email: \texttt{defago@c.titech.ac.jp}}

\author{Yuval Emek}
\affil{Technion, Israel.
Email: \texttt{yemek@technion.ac.il}}

\author{Shay Kutten}
\affil{Technion, Israel.
Email: \texttt{kutten@ie.technion.ac.il}}

\author{Toshimitsu Masuzawa}
\affil{Osaka University, Japan.
Email: \texttt{masuzawa@ist.osaka-u.ac.jp}}

\author{Yasumasa Tamura}
\affil{Tokyo Institute of Technology, Japan.
Email: \texttt{tamura@c.titech.ac.jp}}

\date{}

\begin{document}

\begin{titlepage}

\maketitle

\begin{abstract}
This paper presents a randomized self-stabilizing algorithm that elects a
leader $r$ in a general $n$-node undirected graph and constructs a spanning tree $T$ rooted at $r$.
The algorithm works under the synchronous message passing network model,
assuming that the nodes know a linear upper bound on $n$ and that each edge
has a unique ID known to both its endpoints (or, alternatively, assuming the
$KT_{1}$ model).
The highlight of this algorithm is its superior communication efficiency:
It is guaranteed to send a total of
$\tilde{O} (n)$
messages, each of constant size, till stabilization, while stabilizing in
$\tilde{O} (n)$
rounds, in expectation and with high probability.
After stabilization, the algorithm sends at most one constant size message per
round while communicating only over the
($n - 1$)
edges of $T$.
In all these aspects, the communication overhead of the new algorithm is far
smaller than that of the existing (mostly deterministic) self-stabilizing
leader election algorithms.

The algorithm is relatively simple and relies mostly on known modules that are
common in the fault free leader election literature;
these modules are enhanced in various subtle ways in order to assemble them
into a communication efficient self-stabilizing algorithm.
\end{abstract}


\end{titlepage}

\section{Introduction}
\label{section:introduction}
The \emph{leader election} problem has been recognized early as canonical in
capturing the unique characteristics of distributed systems
\cite{Lann77-DistSystems, GallagerHumbletSpira, lynch1996distributed,
attiya2004distributed}.
Together with related problems, such as \emph{spanning tree} construction and
\emph{broadcast}, it has been extensively studied through the years in various
contexts including
mobile networks
\cite{malpani2000leader, vasudevan2004design},
key distribution
\cite{decleene2001secure},
routing coordination
\cite{perkins1999ad},
sensor control
\cite{heinzelman2000energy},
general control
\cite{hatzis1999fundamental},
Paxos and its practical applications
\cite{lamport2001paxos, lampson1996build, chandra2007paxos, lee1996petal},
peer-to-peer networks
\cite{king2006towards}
and more.
The study of efficient algorithms for leader election and its related
problems still draws plenty of attention in the present, see, e.g.,
\cite{king2015construction, pandurangan2017time, mashreghi2018broadcast,
ghaffari2018distributed, elkin2020simple}.

A central efficiency measure in this regard is the \emph{communication
overhead} that the algorithm adds to the system, in terms of both the number
of messages and their size.
This has been a topic of interest from the early days, e.g., in
\cite{dolev1982n, GallagerHumbletSpira, korach1984tight,
frederickson1987electing, awerbuch1987optimal, afek1991time},
to the more recent literature
\cite{king2015construction, ghaffari2018distributed, elkin2020simple},
in works of theoretical nature as well as in practically motivated ones
\cite{lampson1996build, vasudevan2004design, king2006towards}.

In the realm of \emph{self-stabilizing} algorithms \cite{dijkstra} (see
\cite{dolev2000self, altisen2019introduction}
for textbooks), communication efficiency is a little bit trickier.
Starting from an adversarially chosen initial configuration, the algorithm
sends a certain number of messages until it stabilizes;
after stabilization, the algorithm is required to keep sending messages
indefinitely for the purpose of fault detection (see, e.e.,
\cite{awerbuch1991distributed}).
Consequently, complexity measures related to the algorithm's communication
overhead in the self-stabilization literature are divided into two
``schools'', depending on whether they focus on the messages sent
post-stabilization or pre-stabilization.

The motivation behind bounding the communication overhead after the algorithm
stabilizes comes from the assumption that faults are relatively rare and most
of the time, the system is in a correct configuration.
Here, a natural measure is the algorithm's \emph{stabilization bandwidth},
defined in the influential paper of Awerbuch and Varghese
\cite{awerbuch1991distributed} as the worst case number of messages sent
during any time window of length $\tau$ after the algorithm has stabilized,
where $\tau$ is the algorithm's stabilization time.\footnote{%
In \cite{awerbuch1991distributed}, this notion is measured per edge, dividing
the expression defined in the current paper by the number of edges.}
This measure is justified by the observation that if faults occurred after the
algorithm was supposed to have stabilized, then the algorithm must recover
from them also in $\tau$ time, hence a time window of length $\tau$ inherently
encapsulates a ``full cycle'' of the operations executed by the algorithm
post-stabilization.

The pre-stabilization approach is motivated by the realization that faulty
configurations may lead to bursts of heavy communication, risking an overload
of the system's communication components.
The natural measure in this regard counts the number of messages sent until
the algorithm stabilizes, starting from a worst case initial configuration
(see, e.g., \cite{kutten2010low}).

In this paper, we develop a randomized self-stabilizing algorithm that elects
a leader $r$ and constructs a spanning tree $T$ rooted at $r$ in a general
$n$-node (undirected) communication graph, assuming the \emph{synchronous
$KT_{1}$} model of \cite{awerbuch1990trade} (or a slightly weaker version
thereof, see \Sect{}~\ref{section:model}).
Using constant size messages, the algorithm is guaranteed to stabilize in
$\tilde{O} (n)$
rounds, while sending
$\tilde{O} (n)$
messages, in expectation and whp.\footnote{%
The asymptotic notation
$\tilde{O} (x)$
hides
$\operatorname{polylog}(x)$
factors.
Refer to \Thm{} \ref{theorem:stabilization-time} and
\ref{theorem:message-complexity} for more accurate (asymptotic) bounds.}%
\,\footnote{%
An event $A$ occurs whp (with high probability) if
$\Pr(A) \geq 1 - n^{-c}$
for an arbitrarily large constant $c$.}
The algorithm's stabilization bandwidth is also
$\tilde{O} (n)$
as it is guaranteed to send at most one (constant size) message per round
after stabilization.
Another appealing feature is that after stabilization, the algorithm's
communication is restricted to the edges of $T$.

The communication overhead of the new algorithm significantly improves upon
the state-of-the-art for self-stabilizing leader election algorithms in
general graphs with respect to the aforementioned complexity measures:
To the best of our knowledge, no algorithm in the existing literature gets
below the
$\Omega (m)$
bound for neither the number of messages sent before stabilization, nor the
stabilization bandwidth, where $m$ is the number of edges in the graph.
In fact, the algorithm's communication efficiency matches (up to logarithmic
factors) the state-of-the-art for leader election also in the fault free
setting in terms of both the number of messages and the number of bits sent
before stabilization.

\subsection{Model and Problem}
\label{section:model}
Consider a communication network represented as a simple undirected graph
$G = (V, E)$
whose nodes are identified with processing units that may exchange messages of
constant size with their neighbors.
The execution progresses in synchronous \emph{rounds}, where round
$t \in \Integers_{\geq 0}$
starts at time $t$.
In each round $t$, node
$v \in V$
(1)
receives the messages sent to it over its incident edges in round
$t - 1$
(if any);
(2)
performs local computation;
and
(3)
sends messages over a subset of its incident edges.

Let
$n = |V|$.
Each edge
$e \in E$
admits a unique ID, represented as a bit string of size
$O (\log n)$,
that is known to both endpoints of $e$ (a slightly weaker assumption than that of the $KT_{1}$ model \cite{awerbuch1990trade}).
The nodes also know a linear upper bound $N$ on $n$ that is assumed to be a
sufficiently large power of $2$.

The goal in the \emph{leader election} problem is to reach a configuration
where exactly one node
$r \in V$
is marked as a leader (each node
$u \in V - \{ r \}$
knows that it is not the leader).
In this paper, we also require that the nodes maintain a spanning tree of $G$
rooted at the leader $r$.

We wish to develop a \emph{self-stabilizing} leader election algorithm, where
the initial configuration, at time $0$, is determined by a malicious adversary
that knows the algorithm's code but is oblivious to its random coin tosses (if
any).
When determining the initial configuration, the adversary may set all
variables maintained by node
$v \in V$,
including its local clock (if such a variable is maintained by $v$) and
incoming messages, with the exception of the variables that store the IDs of
$v$'s incident edges and the upper bound $N$ on $n$.

\LongVersion 
\subsection{Additional Related Work and Discussion}
\label{section:related-work}
The shortage of work on communication overhead in the context of
self-stabilization may have resulted from the fact that any non-trivial
self-stabilizing system must send messages infinitely often.
This makes its ``message complexity'' (as defined for non-self-stabilizing
systems) infinite \cite{awerbuch1991distributed}, giving rise to multiple
competing communication measures.
Luckily, the current algorithm is efficient in all the suggested measures.

Another possible explanation is that the known self-stabilizing leader
election algorithms were not developed from the efficient non-self-stabilizing
leader election techniques that preceded them, e.g., they are not derived from
the seminal work of Gallager et al.~\cite{GallagerHumbletSpira}.
This is probably because turning a sophisticated algorithm, designed for a
fault free environment, into a self-stabilizing one is often a complex task
which is prone to mistakes.
Hence, less involved, though communication wasteful, methods were the common
choice for the design of self-stabilizing algorithms.\footnote{%
In explaining how to design an efficient fault tolerant system such as an
efficient implementation of Paxos, Lampson writes that ``More efficiency means
more complicated invariants'' and quotes Dijkstra as saying that ``An
efficient program is an exercise in logical brinkmanship''
\cite{lampson1996build}.}
This contrasts the algorithm developed in the current paper, assembled mainly
from modules that are common in the fault free leader election literature.

Early research on self-stabilization concentrated initially on whether a
self-stabilizing algorithm is at all possible for a given task
\cite{dijkstra, afek1993self, dolev1993self, katz1993self, afek1997local}
and the required memory size (see, e.g.,
\cite{dijkstra, afek1997local, dolev1999memory}).
The main emphasis later has been on reducing the time complexity in various
forms, such as the number of rounds, the asynchronous time, or the number of
steps.

In addition to the aforementioned measures for the communication overhead of a
self-stabilizing algorithm, it was suggested to measure self-stabilization
(and also the related weak detectors) communication efficiency by the number
of links over which the local checking is performed indefinitely after
stabilization, or the maximum number of such links per node
\cite{aguilera2001stable, delporte-gallet2007robust, masuzawa2009communication, masuzawa2009communication-efficient, masuzawa2011silence, datta2014}.
For leader election, the optimum is
$n - 1$
links \cite{larrea2000optimal, delporte-gallet2007robust}.
We note that the algorithm presented in the current paper exactly matches this
optimal bound.

Reducing communication overhead is also mentioned as a motivation for reducing
the number of steps before stabilization, e.g., in \cite{cournier2019first},
where it is also argued that the number of steps in the shared memory model is
closely related to the amount of communication needed to implement that
model.
The step complexity has been addressed in other papers as well, e.g.,
\cite{dijkstra, devismes2016silent, cohen2019first}.

Reducing the communication overhead is also stated as a primary advantage of
silent self stabilizing algorithms \cite{dolev1999memory}.
In algorithms where neighbors exchange the description of their states
periodically (e.g., when using local checking), reducing the state size
translates to a reduced message size.
Approaches for reducing the message size for detection (though not necessarily
the number of messages) by sending only some compressed versions of the state
were studied in
\cite{dolev1991resource, patt2017proof}.
The number of messages in certain algorithms with a single starter (irrelevant
to the leader election task) is reduced in \cite{durand2019reducing}, but is
still
$\Omega (n^{2})$
until stabilization and
$\Omega(n)$
per time unit after stabilization, when applied in synchronous networks.

The tasks of leader election and tree constructions are common building blocks
in numerous algorithms and practical systems such as
mutual exclusion \cite{Lann77-DistSystems},
handling various race conditions and topology dependent tasks
\cite{GallagerHumbletSpira, chandra2007paxos},
ensuring that the components in distributed applications remain consistent
\cite{lamport2019part},
implementing databases and data centers
\cite{isard2007autopilot},
locks
\cite{27897},
file servers
\cite{ghemawat2003google, chang2008bigtable},
broadcast and multicast
\cite{perlman2000interconnections, castro2002scribe},
and topology update and virtually every global task
\cite{awerbuch1990communication}.
In the context of self-stabilization, given a spanning tree, it is possible to
apply a self-stabilizing reset, and to use the reset as a module of a
transformer that can transform non-self-stabilizing algorithms to be
self-stabilizing
\cite{arora1990distributed, awerbuch1991distributed, awerbuch1991self, awerbuch1994self, afek1997local}.
The task of converting leader election algorithms themselves to be
self-stabilizing has not enjoyed the help of such transformers, since most
transformers use a leader and/or a spanning tree.

Numerous self-stabilizing leader election algorithms appeared in the
literature;
we mention here only a few
\cite{afek1990memory, arora1990distributed, awerbuch1991distributed,
dolev1993self, awerbuch1993time, dolev1997superstabilizing,
afek1997self, higham2001self, burman2007time, blin2009new, datta2010self, Alsulaiman2013low, kravchik2013time, Blin2020compact}.
Non constant space is needed for leader election
\cite{beauquier1999memory}
by a deterministic protocol even under a centralized daemon if the nodes
identities ($ID$s) are not bounded.
Logarithmic space is needed if the algorithm is silent
\cite{dolev1999memory}.

The algorithm presented in the current paper utilizes a token circulation and timeouts.
Self-stabilizing token circulation algorithms in general networks have been
treated in numerous papers, e.g.,
\cite{huang1993self, johnen1995distributed, johnen1997self, colette2000,
petit2001fast}.
Multiple papers deal with the related problem of a self-stabilizing
construction of a depth first search tree, starting with
\cite{collin1994self}.
Some self-stabilizing tree construction algorithms are implemented on top of
tokens (or similar messages that are logically sent to nodes that are not immediate neighbors), e.g.,
\cite{higham2001self, blin2009new}.
Timeouts are commonly used in distributed computing, including in leader
election algorithms, e.g.,
\cite{frederickson1987electing, higham2001self, blin2009new}.

The $KT_{1}$ model is advocated in \cite{awerbuch1990trade} as being more
realistic than $KT_{0}$, where a node only knows its ports, but not the node
connected to each port.
In recent years, several non-self-stabilizing algorithms were proposed to
explore the properties of $KT_{1}$ in order to reduce the message complexity
\cite{king2015construction, ghaffari2018distributed, gmyr2018time, mashreghi2018broadcast, mashreghi2019brief}.
\LongVersionEnd 

\subsection{Paper's Organization}
The rest of the paper is organized as follows.
Following some notation and terminology defined in
\Sect{}~\ref{section:prelimineries}, the algorithm is described in
\Sect{}~\ref{section:algorithm}, starting with an informal overview
(\Sect{}~\ref{section:overview}) that includes a discussion of the main
technical ideas.
\Sect{}~\ref{section:analysis} presents the algorithm's analysis, including
its fault recovery guarantees, stabilization run-time, and message
complexity.
\ShortVersion 
Additional related work can be found in the attached full version.
\ShortVersionEnd 

\section{Preliminaries}
\label{section:prelimineries}
Throughout this paper, it is assumed that (directed and undirected) graphs may
include self loops, but not parallel edges.
We denote the node set and edge set of a (directed or undirected) graph $G$ by
$\V(G)$ and $\E(G)$, respectively.

Consider an undirected graph $G$.
A subgraph $T$ of $G$ that admits a tree topology is called a \emph{subtree}
of $G$.
Two node disjoint subtrees $T$ and $T'$ of $G$ are said to be \emph{adjacent}
if there exists an edge
$e = \{ v, v' \} \in \E(G)$
such that
$v \in \V(T)$
and
$v' \in \V(T')$;
such an edge $e$ is referred to as a \emph{crossing} edge.
A \emph{merger} of the subtrees $T$ and $T'$ (over the crossing edge $e$)
forms a new subtree of $G$ whose node set is
$\V(T) \cup \V(T')$
and whose edge set is
$\E(T) \cup \E(T') \cup \{ e \}$.

Consider a directed graph $D$.
The \emph{undirected version} of $D$ is the undirected graph obtained from $D$
by ignoring the edge directions.
We say that $D$ is \emph{weakly connected} if the undirected version of $D$ is
connected (note the distinction from the notion of a strongly connected
directed graph).

The directed graph $D$ is said to be a \emph{pseudoforest} if the outdegree of
every node
$v \in \V(D)$
is at most $1$.
A \emph{pseudotree} is a weakly connected pseudoforest.
The undirected versions of a pseudoforest and a pseudotree are referred to as
an \emph{undirected pseudoforest} and an \emph{undirected pseudotree},
respectively.

We subsequently reserve the notation $G$ for the $n$-node undirected communication graph and denote
$V = \V(G)$
and
$E = \E(G)$.
For a node
$v \in V$,
its neighbor set in $G$ is denoted by
$\Neighbors(v) = \{ u \in V \mid \{ u, v \} \in E \}$.

\section{The Algorithm}
\label{section:algorithm}

\subsection{An Informal Overview}
\label{section:overview}
In this section, we provide an overview of our algorithm, composed of various
modules that the informed reader will recognize from the vast literature on
leader election and spanning tree construction in fault free environments.
The algorithm maintains a partition of the nodes into rooted trees, defined
by means of variables $v.\varParent$ and $v.\varChildren$ in which each node
$v$ stores its tree parent and children, respectively.
The actions of each tree $T$ are coordinated by its root, namely, the unique
node
$r \in \V(T)$
satisfying
$r.\varParent = \bot$.
The goal is to repeatedly merge the trees over crossing edges until a
single tree that spans the whole graph remains.

The algorithm's main module, denoted by \ModuleMain{}, divides the
execution into \emph{phases}, where in each phase the root $r$ of tree $T$
decides at random between two modules, denoted by \ModuleProposing{} and
\ModuleAccepting{}, to be invoked in the current phase.
The role of \ModuleProposing{} is to find a crossing edge over which a merger
proposal is sent.
To this end, $r$ invokes a randomized module, denoted by
\ModuleFindCrossingEdge{}, that implements procedure $\mathit{FindAny}$ of
\cite{king2015construction}.
If \ModuleFindCrossingEdge{} fails to find a crossing edge, then the phase
ends.

Otherwise, an edge
$\{ x, y \}$,
that crosses between
$\V(T) \ni x$
and
$V - \V(T) \ni y$,
is returned.
Then, $r$ invokes a module, denoted by \ModuleRootTransfer{}, whose role is to
update the $\varParent$ and $\varChildren$ variables along the unique simple
$(r, x)$-path
in $T$ so that $T$ is re-rooted at $x$ (cf.\ \cite{GallagerHumbletSpira}).
Following that, $x$ sends a merger proposal to $y$ and waits silently for a
reply.
If such a reply does not arrive within the next
$\Theta (N)$
rounds, then the phase ends and $x$, now being the root, initiates a new
phase.
Otherwise, $x$ becomes the child of $y$, merging $T$ into $y$'s tree.

The role of \ModuleAccepting{} is to accept incoming merger proposals.
In particular, a merger proposal sent from a node $x$ to a node $y$ in tree
$T'$ is accepted by $y$, sending a reply to $x$, if and only if the current
phase of $T'$ (``as far as $y$ knows'') is an \ModuleAccepting{} phase.
This ensures that the resulting structure is cycle free.

The aforementioned process is implemented on top of a module, denoted by
\ModuleTraversal{}, that implements a depth first search traversal of the tree
by a conceptual \emph{token}.
The executions of both \ModuleFindCrossingEdge{} and \ModuleAccepting{} are
divided into
$\Theta (\log N)$
\emph{epochs} so that \ModuleTraversal{} is invoked by the root $r$ at the
beginning of each epoch.
The epochs lasts for a fixed
$\Theta (N)$
number of rounds, chosen to guarantee that the traversal can be safely
completed, returning the token to $r$ where it is stored until the next epoch
begins.
Apart from \ModuleTraversal{} that controls the token's mobility,
\ModuleRootTransfer{} also shifts the token down the root transfer path, thus
ensuring that outside the scope of the traversals handled by
\ModuleTraversal{}, the token is always stored at the root.

A key property of the algorithm, that facilitates its communication
efficiency, is that a node may send a message only when it holds the token
(and then, it may send at most one message per round).
In particular, \ModuleFindCrossingEdge{} implements
procedure $\mathit{FindAny}$ \cite{king2015construction} on top of the tree
traversals, so that each one of its
$\Theta (\log N)$
epochs is responsible for one broadcast-echo process in the tree (carrying
$O (1)$
bits of information).
Moreover, a merger proposal sent from node $x$ to node $y$ is recorded at
$y$ (assuming that $y$'s tree is in an \ModuleAccepting{} phase) until $y$
holds the token as part of a traversal of its tree;
the proposal is then processed and $y$ sends a reply to $x$.
The token held by $x$ prior to the merger is dissolved when $x$'s tree is
merged into $y$'s, striving for the invariant that each tree holds a single
token.

The fact that each tree is repeatedly traversed by its token is exploited by
the algorithm's fault detection mechanism:
Each node $v$ maintains a $v.\varTimer$ variable that is reset when $v$
receives the token from its parent as part of a traversal and is incremented
in every round otherwise;
if $v.\varTimer$ exceeds a predetermined
$\Theta (N)$
threshold, then $v$ invokes a procedure named \ProcRestart{} that resets all
its variables, so that $v$ forms a new singleton tree, and generates a fresh
token at $v$.
A malformed tree is (implicitly) detected by the token being lost when it is
sent to $v$ from an adjacent node $u$ while $v$ does not ``expect'' to receive
a token from $u$.
We emphasize that in both events ($v$ experiencing a restart and $v$ ignoring
$u$'s message), $v$ does not inform any of its neighbors of the detected
fault.

This fault detection mechanism essentially replaces the local checking module,
common to many self-stabilizing algorithms:
rather than checking ``all neighbors all the time'', node $v$ checks only the
tokens it receives, which is clearly more efficient communication-wise.
Our algorithm's self-stabilization guarantees follow from this mechanism with
some additional fine points explained in
\Sect{}~\ref{section:detailed-description}.
\LongVersion 
Here, we mention just some specific subtleties:
\begin{itemize}

\item
The phase and epoch durations are carefully calculated such that each time
tree $T$ initiates a phase, it has a constant probability to merge into a
neighboring tree, even though the phases of the trees may be (adversarially)
shifted in time.

\item
To detect the erroneous case that a tree has two tokens that pass each other
``in the dark'' over an edge, a node refuses to receive a token immediately
after it sent one.

\item
A restarting node does not interact with any other node for
$\Omega (N \log N)$
rounds.

\item
After replying to a merger proposal from node $x$, node $y$ adds $x$ to
$y.\varChildren$ (as part of the merger) and continues the traversal by
passing the token to $x$, thus ensuring that the newly annexed tree is
traversed before the current traversal terminates.

\item
During an \ModuleAccepting{} (resp., \ModuleProposing{}) phase in tree $T$,
the token updates $T$'s nodes of the phases' type in a bottom up (resp.,
top down) order.

\item
The path from the root $r$ to the crossing edge returned by
\ModuleFindCrossingEdge{} is marked, using designated $\varOutProposal$
variables, only in the last of the
$\Theta (\log N)$
epochs of the module;
this is done in a bottom up order in a manner ensuring that
$r.\varOutProposal \neq \bot$
if and only if the desired path is indeed marked.

\item
The epoch duration is chosen so that the token is held by the root for
$\Theta (N)$
rounds after each traversal ends, which means that the rounds in which a node
receives the token from its parent are supposed to be
$\Theta (N)$
spaced apart;
\ProcRestart{} is invoked if they are not.

\end{itemize}
\LongVersionEnd 

\subsection{A Detailed Description}
\label{section:detailed-description}
Our algorithm performs leader election by constructing a rooted spanning tree
of $G$.
To this end, each node
$v \in V$
maintains a variable
$v.\varParent \in \Neighbors(v) \cup \{ \bot \}$
that points to its tree parent and a variable
$v.\varChildren \in (\Neighbors(v))^{*}$
that stores a list of pointers to its tree children;
by a slight abuse of notation, we may address $v.\varChildren$
as a subset of $\Neighbors(v)$ when the order of its elements is not
important.

We say that $v$ is a \emph{root} if
$v.\varParent = \bot$.
Let
$v.\varTreeNeighbors = v.\varChildren$
if $v$ is a root;
and
$v.\varTreeNeighbors = v.\varChildren \cup \{ v.\varParent \}$
if $v$ is not a root.
We refer to the nodes in $v.\varTreeNeighbors$ as the \emph{tree neighbors} of
$v$ (from the perspective of $v$).

For the sake of simplifying the algorithm's description, we augment the
communication graph with a virtual \emph{shadow} node $\tilde{v}$ for each
node $v$ of the original graph so that $\tilde{v}$ has a single incident edge
connecting it to $v$.
The algorithm is designed so that $v$ is always the parent of $\tilde{v}$
which guarantees, in particular, that the set of tree neighbors of a node is
never empty.
To distinguish the original graph nodes from their virtual shadows we refer to
the former as \emph{physical} nodes.
The operation of the virtual shadow node $\tilde{v}$ is simulated by its
corresponding physical node $v$;
this simulation is straightforward as $v$ is the only neighbor of $\tilde{v}$.
We assume hereafter that the node set $V$ of the communication graph
$G = (V, E)$
contains both the physical and shadow nodes.

\LongVersion 
The aforementioned modules
\ModuleTraversal{}, \ModuleMain{}
\ModuleProposing{}, \ModuleAccepting{}, \ModuleFindCrossingEdge{}, and
\ModuleRootTransfer{} are presented in \Sect{}\
\ref{section:module-traversal}--\ref{section:module-root-transfer},
respectively, while \ProcRestart{} is presented in
\Sect{}~\ref{section:procedure-restart}.
\LongVersionEnd 
\ShortVersion 
The aforementioned modules \ModuleTraversal{}, \ModuleMain{}
\ModuleProposing{}, and \ModuleAccepting{} are presented in \Sect{}\
\ref{section:module-traversal}--\ref{section:module-accepting},
respectively;
due to lack of space, the description of \ModuleRootTransfer{} and
\ProcRestart{}, whose implementation is fairly straightforward, is deferred to
the attached full version.
\ShortVersionEnd 
To help the reader put things into context, some parts of the algorithm's
description are written as if no ``recent faults'' has occurred;
we emphasize that the correctness of the algorithm does not rely on any such
assumption.
Moreover, to clarify the algorithm's presentation, we restrict our attention
to the less trivial components, leaving out some details that the reader can
easily complete by themselves;
for ease of reference, Table~\ref{table:variables} summarizes the variables
maintained by these components.

\subsubsection{Module \ModuleTraversal{}}
\label{section:module-traversal}
Consider some root node $r$ and let $T$ be $r$'s tree.
The \ModuleTraversal{} module implements a distributed process in which a
conceptual \emph{token} is passed from node to node, using designated
\msgPassToken{} messages, to form a depth first search traversal of $T$.

To monitor the token distribution, each node
$v \in V$
maintains three variables:
(a)
$v.\varToken \in \{ 0, 1 \}$
that indicates whether $v$ holds a token;
(b)
$v.\varTokenDirection \in v.\varTreeNeighbors$
that points to the last tree neighbor to which $v$ has passed a token;
and
(c)
$v.\varRecentPass \in \{ 0, 1 \}$
that indicates whether $v$ has passed a token in the previous round.

Denoting
$d = |v.\varTreeNeighbors|$,
node $v$ also employs a permutation (i.e., a bijection)
$\pi_{v} : \{ 0, 1, \dots, d - 1 \} \rightarrow v.\varTreeNeighbors$
defined so that
$\pi_{v}(i - 1)$
points to the $i$-th element in $v.\varChildren$ for
$1 \leq i \leq |v.\varChildren|$
and
$\pi_{v}(d - 1)$
points to $v.\varParent$ if
$v.\varParent \neq \bot$.
This permutation naturally induces a periodic sequence over the nodes in
$v.\varTreeNeighbors$ so that the $i$-th element of this sequence is
$\pi_{v}(i \bmod d)$
for every
$i \in \Integers_{> 0}$,
referring to
$\pi_{v}(i + 1 \bmod d)$
as its \emph{$\pi_{v}$-successor}.
If $v$ is a physical node, which ensures that
$v.\varChildren \neq \emptyset$,
then we refer to $\pi_{v}(0)$ and to
$\pi_{v}(|v.\varChildren| - 1)$
as $v$'s \emph{first child} and \emph{last child}, respectively.

Consider some non-root node $v$.
Upon receiving a \msgPassToken{} message $M$ from node
$u \in \Neighbors(v)$,
node $v$ verifies that
(i)
$v.\varToken = 0$;
(ii)
$v.\varRecentPass = 0$;
(iii)
$v.\varTokenDirection = u$;
and
(iv)
its parent-child relations with $u$ are consistent between the two nodes, that
is, either
$u = v.\varParent$
and $M$ indicates that $u$ is the parent of $v$ or
$u \in v.\varChildren$
and $M$ indicates that $u$ is a child of $v$.
If any of these four conditions is not satisfied, then $v$ ignores $M$
altogether, thus causing the token it carries to (conceptually) disappear.

Assuming that the four conditions are satisfied, node $v$ sets
$v.\varToken \gets 1$,
thus indicating that it now holds the token that conceptually arrived with
message $M$.
Following that, $v$ holds the token for $1$--$4$ full rounds, where the exact
number is determined by the modules implemented on top
of \ModuleTraversal{}.
Then, $v$ passes the token to the $\pi_{v}$-successor $u'$ of
$u = v.\varTokenDirection$
by sending to it a \msgPassToken{} message and sets
(1)
$v.\varToken \gets 0$;
(2)
$v.\varTokenDirection \gets u'$;
and
(3)
$v.\varRecentPass \gets 1$.
The sole purpose of the flag $v.\varRecentPass$ is to ensure that $v$ holds no
token for (at least) one full round before it can accept a token again.
This flag is always turned off one round after it has been turned on;
that is, $v$ resets
$v.\varRecentPass \gets 0$
in each round, after the value of $v.\varRecentPass$ has been read as part of
processing the incoming messages.

The operation of the root $r$ (which is always a physical node) in
\ModuleTraversal{} is identical to that of any non-root node except that $r$
is also responsible for initiating the traversal, by passing the token to its
first child, and terminating the traversal, after receiving the token from its
last child, holding the token in between traversals.
To distinguish the situation where the token is held by a (root or non-root)
node in the midst of a traversal from the situation in which the token is held
by a (root) node between traversals, we refer to the token in the former
situation as \emph{hot} and in the latter situation as \emph{cold}.
The traversal process is initiated at $r$ upon receiving a designated signal
from the modules implemented on top of \ModuleTraversal{}.
If $r$ does not hold a cold token when this signal is received, then
\ProcRestart{} is invoked;
otherwise, the token becomes hot and a traversal of $T$ is initiated.

Consider some traversal of $T$.
We refer to the \msgPassToken{} message received (resp., sent) by a non-root
node $v$ from (resp., to) its parent as $v$'s \emph{discovery message} (resp.,
\emph{retraction message}) and to the round in which this message is received
(resp., sent) as $v$'s \emph{discovery round} (resp., \emph{retraction
round}).
The \emph{discovery round} (resp., \emph{retraction round}) of $T$'s root $r$
is defined to be the round in which the traversal is initiated (resp.,
terminated), i.e., the round in which the token held by $r$ turns from cold to
hot (resp., from hot to cold).
A key observation is that the traversal allows $r$ to simulate a
broadcast-echo process over its tree $T$ by piggybacking the content of the
broadcast (resp., echo) messages over the discovery (resp., retraction)
messages.

\ModuleTraversal{} is responsible also for a timer mechanism monitored by
variable
$v.\varTimer \in \Integers_{\geq 0}$
that each node
$v \in V$
maintains.
This variable is incremented by $v$ in every round.
During a discovery round, $v$ verifies that
$v.\varTimer \geq \Ctr N$,
where $\Ctr$ is a positive constant whose value is determined in
\Sect{}~\ref{section:analysis}.
If this condition is not satisfied, an event referred to as a \emph{premature
discovery}, then $v$ invokes \ProcRestart{};
otherwise, $v$ resets
$v.\varTimer \gets 0$.
If, at any stage of the execution, variable $v.\varTimer$ exceeds the
$8 \Ctr N$
threshold, then $v$ also invokes \ProcRestart{};
this event is referred to as an \emph{expiration} of $v.\varTimer$.

\subsubsection{Module \ModuleMain{}}
\label{section:module-main}
The execution of \ModuleMain{} is partitioned into \emph{phases} that
are further partitioned into \emph{epochs};
these partitions are orchestrated by the root nodes $r$.
Each phase is dedicated to an invocation of either \ModuleProposing{} or
\ModuleAccepting{};
the root $r$ decides between the two at the beginning of the phase by tossing
a fair coin.

A \ModuleProposing{} phase consists of 
$\Cph \log N$
\emph{search} epochs,
where
$\Cph \in \Integers_{> 0}$
is a constant to be determined in the sequel,
followed by a \emph{root transfer} epoch, followed by a \emph{proposing}
epoch.
The search epochs are dedicated to an invocation of \ModuleFindCrossingEdge{}
and based on its outcome, the root transfer epoch is either empty, in
which case, the proposing epoch is also empty, or dedicated to
an invocation of \ModuleRootTransfer{}.
An \ModuleAccepting{} phase consists of
$\Cph \log N + 2$
\emph{accepting} epochs.

Consider a tree $T$ rooted at $r$.
Each search and accepting epoch is dedicated to an invocation of
\ModuleTraversal{}, initiated by $r$, during which the (hot) token completes
one traversal of $T$ and returns to $r$, where it is held (cold) until the
epoch ends.
These epochs have a fixed length of
$2 \Ctr N$
rounds.
The length of the root transfer and proposing epochs may vary, however,
it is guaranteed that their combined length is at most
$4 \Ctr N$
rounds.
Therefore, the \ModuleAccepting{} phases last for exactly
$(\Cph \log N + 2) \cdot 2 \Ctr N$
rounds, whereas the length of each \ModuleProposing{} phase is at least
$\Cph \log N \cdot 2 \Ctr N$
and at most
$(\Cph \log N + 2) \cdot 2 \Ctr N$
rounds.

A variable that plays an important role in the distinction between
\ModuleAccepting{} and \ModuleProposing{} phases is
$v.\varAcceptingProposals \in \{ 0, 1 \}$,
maintained by each node
$v \in \V(T)$,
that indicates weather $v$ is available for merger proposals (more on that
soon).
Node $v$ turns this variable on (if it is not on already) in the retraction
round of any traversal associated with an accepting epoch (i.e., during
\ModuleAccepting{} phases);
it turns this variable off (if it is not off already) in the discovery
round of any traversal associated with a search epoch (i.e., during
\ModuleProposing{} phases).
This means that the $\varAcceptingProposals$ variables are turned on in a
bottom up order and turned off in a top down order.
Notice that once the variable $v.\varAcceptingProposals$ is on (resp., off),
it stays on (resp., off) at least until the end of the current
\ModuleAccepting{} (resp., \ModuleProposing{}) phase.

\subsubsection{Module \ModuleProposing{}}
\label{section:module-proposing}
Consider some tree $T$ rooted at node $r$.
\ModuleProposing{} controls the algorithm's operation during
\ModuleProposing{} phases orchestrated by $r$.
Specifically, during the
$\Cph \log N$
search epochs of a \ModuleProposing{} phase, $r$ orchestrates the search for
an edge that crosses from $T$ to an (arbitrary) adjacent tree.
This is done by invoking \ModuleFindCrossingEdge{} that is guaranteed to
return a crossing edge with a positive constant probability if such an edge
exists.

The output of \ModuleFindCrossingEdge{} is returned by means of the variables
$v.\varOutProposal \in \Neighbors(v) \cup \{ \bot \}$
that each node
$v \in \V(T)$
maintains.
Assuming that no faults have occurred during the execution of
\ModuleFindCrossingEdge{}, it is guaranteed that these variables satisfy the
following three properties:
(OP1)
If
$r.\varOutProposal = \bot$,
then
$v.\varOutProposal = \bot$
for every node
$v \in \V(T)$.
(OP2)
If
$r.\varOutProposal \in \Neighbors(v)$,
then there exists exactly one \emph{crossing port}
$x \in \V(T)$,
namely, a node satisfying
$x.\varOutProposal \in \Neighbors(x) - x.\varTreeNeighbors$.
Taking $P$ to be the unique simple
$(r, x)$-path
in $T$, the variable $v.\varOutProposal$ points to the successor of $v$ along
$P$ for every node
$v \in \V(P) - \{ x \}$;
and
$v.\varOutProposal = \bot$
for every node
$v \in \V(T) - \V(P)$.
(OP3)
If $x$ is a crossing port, then
$x.\varOutProposal = y \in V - \V(T)$.
Refer to Figure~\ref{figure:tree-merge} for an illustration.

If \ModuleFindCrossingEdge{} does not find a crossing edge, indicated by
$r.\varOutProposal = \bot$,
then the phase terminates (which means that the root transfer and proposing
epochs are empty).
Assuming that \ModuleFindCrossingEdge{} returns the edge
$\{ x, y \}$
that connects the crossing port
$x \in \V(T)$
with a node
$y \in V - \V(T)$,
the root transfer epoch is dedicated to modifying the parent-child relations
in $T$ so that the root role transfers from $r$ to $x$;
this is handled by \ModuleRootTransfer{}.

When \ModuleRootTransfer{} terminates, the tree $T$ is rooted at $x$, that also
holds the (cold) token, and
$x.\varOutProposal = y$.
The proposing epoch is now dedicated to the attempt of $x$ to merge
$T$ and $y$'s tree $T'$ over the crossing edge
$e = \{ x, y \}$.
To this end, in the first round of the proposing epoch, $x$ sends a
\msgPropose{} message to $y$.
Following that, $x$ waits silently for a reply from $y$ in the form of an
\msgAccept{} message.
If this reply does not arrive during the subsequent
$3 \Ctr N - 1$
rounds, then $x$ sets
$x.\varOutProposal \gets \bot$
and the proposing epoch terminates (with no tree merger),
also terminating this \ModuleProposing{} phase (the next phase starts with $x$
as the root of $T$).

If $x$ receives an \msgAccept{} message from $y$ before the end of the
proposing epoch, then $T$ is merged with $T'$.
From the perspective of $x$, this merger is executed by setting
(1)
$x.\varParent \gets y$;
(2)
$x.\varToken \gets 0$;
(3)
$x.\varTokenDirection \gets y$;
and
(4)
$x.\varOutProposal \gets \bot$.
In particular, this means that $x$ is no longer a root and that it expects the
next \msgPassToken{} message to arrive from its new parent $y$.
Moreover, the token held by $x$ is eliminated, an event referred to hereafter
as \emph{dissolving} $x$'s token.

\subsubsection{Module \ModuleAccepting{}}
\label{section:module-accepting}
Consider some tree $T'$ rooted at $r'$.
\ModuleAccepting{} controls the algorithm's operation during
\ModuleAccepting{} phases orchestrated by $r'$.
The role of an \ModuleAccepting{} phase is to allow the nodes of $T'$ to
accept merger proposals of neighboring nodes in adjacent trees.
To this end, each node
$y \in \V(T')$
maintains the binary variable
$y.\varInProposals(x) \in \{ 0, 1 \}$
for every (graph) neighbor
$x \in \Neighbors(y) - y.\varTreeNeighbors$.
This variable registers merger proposals received from $x$ so that they can be
processed when $y$ is ready (as explained soon).
Specifically, $y.\varInProposals(x)$ is turned on when $y$ receives a
\msgPropose{} message while
$y.\varAcceptingProposals = 1$;
it is turned off when $y$ resets
$y.\varAcceptingProposals \gets 0$.
We emphasize that any \msgPropose{} message received while
$y.\varAcceptingProposals = 0$
is ignored.
Moreover, if a node
$x \in \Neighbors(y)$
joins $y.\varTreeNeighbors$, then the $y.\varInProposals(x)$ entry in
$y.\varInProposals(\cdot)$ is eliminated.

An \ModuleAccepting{} phase consists of
$\Cph \log N + 2$
accepting epochs, each one of them begins by signaling $r'$ to invoke
\ModuleTraversal{}, thus initiating a traversal process during which a (hot)
token is passed through the nodes of $T'$;
when the traversal terminates, the token is held (cold) at $r'$ until the next
epoch.
A physical node
$y \in \V(T')$
is said to be \emph{retraction ready} if it holds a hot token while
$y.\varTokenDirection = \pi_{y}(|y.\varChildren| - 1)$,
indicating that the held token has been received from $y$'s last child.
The algorithm is designed so that $y$ may respond to merger proposals of
its neighbors only when it is retraction ready.
Notice that a shadow node $\tilde{y}$ never receives merger proposals as its
only neighbor $y$ in $G$ is always a tree neighbor of $\tilde{y}$.

Consider node $y$ when it becomes retraction ready while
$y.\varAcceptingProposals = 1$
and let $w$ be the (current) last child of $y$.
Node $y$ checks the content of
$P_{y} = \{ x \mid y.\varInProposals(x) = 1 \}$;
if
$P_{y} = \emptyset$,
then the role of $y$ in the current accepting epoch is over and
\ModuleAccepting{} signals \ModuleTraversal{} that $y$ can release the hot
token (after holding it for one full round), which leads to $y$'s retraction
round.

If
$P_{y} \neq \emptyset$,
then $y$ picks one neighbor
$x \in P_{y}$
(arbitrarily) and sends to it an \msgAccept{} message.
Then, in the subsequent round, the merger of $T'$ and the tree $T$ of $x$ is
executed, where on the $T'$ side, $y$ adds $x$ as its last child in $T'$ by appending $x$ to the end of $y.\varChildren$, turning it into the $\pi_{y}$-successor of (the original last child) $w$.
Following that, \ModuleAccepting{} signals \ModuleTraversal{} that $y$ can
release the hot token, continuing with the traversal of the updated $T'$;
this results in passing the hot token to $x$ as $x$ is now the
$\pi_{y}$-successor of
$w = y.\varTokenDirection$.
Refer to Figure~\ref{figure:tree-merge} for an illustration.

Observe that the tree merger is designed so that $T$ is traversed by the
token held by $y$ immediately after it is merged into $T'$.
The aforementioned mechanism that controls the $\varAcceptingProposals$
variables guarantees that $v.\varAcceptingProposals$ is turned on during
this traversal for every node
$v \in \V(T)$;
this happens during $v$'s retraction round which means that $v$ will be able
to process incoming merger proposals only during the next accepting epoch (if
any).
Node $y$ continues processing merger proposals of other nodes when it gets the
token back from $x$.

\subsubsection{Module \ModuleFindCrossingEdge{}}
\label{section:module-find-crossing-edge}
King et al.~\cite{king2015construction} developed a distributed procedure
called $\mathit{FindAny}$, working in a fault free environment, that given a
subtree $T$ of $G$ rooted at $r$, finds an edge that crosses between $\V(T)$
and
$V - \V(T)$
with a positive constant probability if such an edge exists.
The procedure is orchestrated by $r$ using a constant number of broadcast-echo
iterations over $T$ with messages of size
$O (\log n)$.\footnote{%
The procedure needed for our purposes is a slightly simplified version of the
one developed in \cite{king2015construction}.}

We use a variant of the procedure of King et al., where each broadcast-echo
iteration with
$O (\log n)$
size messages is subdivided into
$O (\log n)$
broadcast-echo iterations with constant size messages.
\ModuleFindCrossingEdge{} is then implemented using
$O (\log n)$
iterations of the traversal process managed by \ModuleTraversal{},
piggybacking the (constant size) messages of each broadcast-echo iteration over the \msgPassToken{} messages of the corresponding traversal process (see \Sect{}~\ref{section:module-traversal}).
Each one of these
$O (\log n)$
traversal iterations is executed in its own search epoch, where $r$ is
signaled to invoke \ModuleTraversal{} at the beginning of the epoch.
The value of $\Cph$ is derived from the hidden constant in the $O$-notation.

Unfortunately, we cannot guarantee that properties (OP1)--(OP3) of the
$\varOutProposal$ variables in $T$ (see
\Sect{}~\ref{section:module-proposing}) hold if these variables sprout from
the adversarially chosen initial configuration, rather than being generated by
a complete invocation of \ModuleFindCrossingEdge{} as part of the execution.
However, we implement the module so that property (OP2) is guaranteed as
long as the last search epoch, referred to hereafter as the \emph{safety}
epoch of \ModuleFindCrossingEdge{}, is completed as part of the execution.

To this end, we implement \ModuleFindCrossingEdge{} so that variable
$v.\varOutProposal$ is set during $v$'s retraction round of the traversal
associated with the safety epoch for every node
$v \in \V(T)$.
Moreover, the retraction message of each child
$u \in v.\varChildren$
of $v$ specifies whether the number of crossing ports in $u$'s subtree is
zero, one, or at least two;
node $v$ then sets
$v.\varOutProposal \gets u$
(in its own retraction round) if and only if $u$ admits one crossing port in
its subtree and every other child of $v$ admits zero crossing ports in its
subtree.

We emphasize that a fault free execution of the safety epoch alone does not
guarantee properties (OP1) and (OP3);
that is, it does not rule out a scenario where
$r.\varOutProposal = \bot$
and yet,
$v.\varOutProposal \neq \bot$
for some non-root nodes
$v \in \V(T)$,
nor does it rule out a scenario where
$r.\varOutProposal \neq \bot$,
but the returned edge $e$ connects the (unique) crossing port $x$ to another
node in $T$ (both scenarios require the adversary's ``involvement'').
We show in the sequel that the latter scenario does not affect the correctness
of our algorithm;
to overcome the former scenario, every node $v$ resets
$v.\varOutProposal \gets \bot$
upon receiving a traversal discovery message of \ModuleTraversal{}, that is,
when $v$ receives a \msgPassToken{} message from $v.\varParent$ (and does not
ignore it).
This is consistent with the fact that \ModuleRootTransfer{} is implemented
using \msgRootTransfer{} messages, rather than \msgPassToken{} messages.

\LongVersion 
\subsubsection{Module \ModuleRootTransfer{}}
\label{section:module-root-transfer}
Given a tree $T$ rooted at $r$ and a crossing port $x$, \ModuleRootTransfer{}
modifies the parent-child relations in $T$ so that it is re-rooted at $x$.
When the module is invoked, the $\varOutProposal$ variables maintained by
the nodes in $T$ mark the unique simple
$(r, x)$-path
$P$ in $T$ (assuming that these variables were generated during the safety epoch of \ModuleFindCrossingEdge{}, completed as part of the execution).
\ModuleRootTransfer{} sequentially shifts the root role from node to node down
the path $P$, using \msgRootTransfer{} messages that carry a cold token.
This is done as follows.

Consider some intermediate node
$v \in \V(P)$,
as indicated by
$v.\varParent \neq \bot$
and
$v.\varOutProposal \in v.\varChildren$.
When $v$ receives a \msgRootTransfer{} message $M$ from node
$u \in \Neighbors(v)$,
it verifies that
(i)
$v.\varToken = 0$;
(ii)
$v.\varRecentPass = 0$;
(iii)
$v.\varTokenDirection = u$;
and
(iv)
$u = v.\varParent$;
if any of these four conditions is not satisfied, then $v$ ignores $M$
altogether, thus causing the token that $M$ carries to (conceptually)
disappear.

Assuming that the four conditions are satisfied, node $v$ sets
$v.\varToken \gets 1$,
indicating that it now holds the (cold) token, and turns itself into the new
root by setting
$v.\varParent \gets \bot$
and appending $u$ to $v.\varChildren$.
(Note that $v$ does not reset variable $v.\varTimer$ upon receiving a cold
token.)
In the subsequent round, $v$ sends a \msgRootTransfer{} message to its
successor
$w = v.\varOutProposal$
in $P$ and sets
(1)
$v.\varToken \gets 0$;
(2)
$v.\varTokenDirection \gets w$;
(3)
$v.\varOutProposal \gets \bot$;
and
(4)
$v.\varRecentPass \gets 1$;
thus indicating that $v$ no longer holds the token and that it expects the next
\msgPassToken{} message to arrive from $w$.
It also turns itself into a child of $w$ by setting
$v.\varParent \gets w$
and
$v.\varChildren \gets v.\varChildren - \{ w \}$;
to ensure that the parent-child relations between $v$ and $w$ are updated in synchrony, these two operations are performed by $v$ in the round subsequent to sending the \msgRootTransfer{} message.

The original root $r$ behaves just like an intermediate node except that its
actions are triggered by the invocation of \ModuleRootTransfer{}, rather than
by receiving a \msgRootTransfer{} message from its (non-existent) predecessor
in $P$.
The crossing port $x$ also behaves just like an intermediate node except that
it remains the (cold token holding) root and does not send a
\msgRootTransfer{} message to its (non-existent) successor in $P$;
rather, it initiates a proposing epoch as explained earlier.
\LongVersionEnd 

\LongVersion 
\subsubsection{Procedure \ProcRestart{}}
\label{section:procedure-restart}
When \ProcRestart{} is invoked at a physical node
$v \in V$
or at its shadow $\tilde{v}$, we say that both $v$ and $\tilde{v}$
\emph{experience} a restart.
Upon invocation of the procedure at either of the two nodes, \ProcRestart{}
disconnects $v$ from all its tree neighbors other than $\tilde{v}$ and places
$v$ and $\tilde{v}$ in a new tree rooted at $v$ that includes only these two
nodes.
This is implemented by setting
$v.\varParent \gets \bot$,
$\tilde{v}.\varParent \gets v$,
$v.\varChildren \gets \{ \tilde{v} \}$,
and
$\tilde{v}.\varChildren \gets \emptyset$.
\ProcRestart{} assigns the new tree's (cold) token to $v$, setting
$v.\varToken \gets 1$,
$\tilde{v}.\varToken \gets 0$,
and
$v.\varRecentPass, \tilde{v}.\varRecentPass \gets 0$,
and consistently adjusts the $\varTokenDirection$ variables by setting
$v.\varTokenDirection \gets \tilde{v}$
and
$\tilde{v}.\varTokenDirection \gets v$.
The $\varOutProposal$ and $\varInProposals$ variables of $v$ and $\tilde{v}$
are initialized, setting
$v.\varOutProposal, \tilde{v}.\varOutProposal \gets \bot$
and
$v.\varInProposals(u) \gets 0$
for every
$u \in \Neighbors(v) - \{ \tilde{v} \}$
(recall that $v$ is the only graph neighbor of $\tilde{v}$ and is also a tree
neighbor of $\tilde{v}$, so $\tilde{v}.\varInProposals(\cdot)$ is ``empty''),
thus indicating that $v$ and $\tilde{v}$ do not participate in any active root
transfer path and that no tree merger proposals are currently registered at
them.
Finally, the $\varTimer$ variables are set to
$v.\varTimer, \tilde{v}.\varTimer \gets \Ctr N$,
thus allowing the nodes to start a traversal (without experiencing another
restart due to a premature discovery), and $v$ initiates a new
\ModuleProposing{} phase.

On top of the fault detection mentioned earlier, \ProcRestart{} is also
invoked whenever the node's internal variables are inconsistent with each
other.
This can happen only in the initial configuration.
\LongVersionEnd 

\section{Analysis}
\label{section:analysis}

\subsection{Notation, Terminology, and Basic Observations}
\label{section:analysis:preliminaries}
We analyze the operation of the algorithm on the communication graph
$G = (V, E)$
that contains both the physical nodes and their (virtual) shadow nodes (see
\Sect{}~\ref{section:algorithm}).
Recalling that a node can send a message only when it holds a token, the
algorithm's proof of correctness relies on analyzing the token distribution
across $G$ and its dynamics over time.
To this end, we regard the conceptual tokens as actual entities that are
transitioned across the graph.

Throughout this section, we append a superscript $t$ to the notation of our
variables when referring to the value that these variables hold at time
$t \in \Integers_{\geq 0}$
so $v.\varDummy^{t}$ denotes the value held by variable $\varDummy$ of node
$v \in V$
at time $t$.
The superscript $t$ is omitted when $t$ is not important or clear from the
context.

\subparagraph{Token Transition.}
Consider a node
$v \in V$
and suppose that
$v.\varToken^{t} = 1$
which means that $v$ holds a token $\kappa$ at time $t$.
This token may be \emph{dispatched} in round $t$ from $v$ to a node
$w \in v.\varTreeNeighbors$
by means of a \msgPassToken{} message if $\kappa$ is hot;
or a \msgRootTransfer{} message if $\kappa$ is cold.
We say that the dispatch \emph{fails} and that $\kappa$ \emph{dies} if this
message is ignored by $w$ (in round
$t + 1$);
otherwise, we say that the dispatch \emph{succeeds} and that $\kappa$ is
\emph{acquired} by $w$.
Notice that in the latter case, we regard the dispatch as successful, saying
that $w$ acquires $\kappa$, even if $w$ experiences a restart in round
$t + 1$
(``after'' acquiring the token).

Node
$v \in V$
is said to \emph{own} a token at time
$t > 0$
if either
(1)
$v.\varToken^{t} = 1$;
or
(2)
$v.\varToken^{t} = 0$
and
$v.\varRecentPass^{t} = 1$.
In other words, $v$ owns a token at time $t$ if it holds a token at time $t$
or if it has just dispatched a token to a tree neighbor $u$ in round
$t - 1$
(which means that $v$ no longer holds it at the beginning of round $t$).
Using this convention, we ensure that if the dispatch of the token from $v$ to
$u$ is successful, then the token is delivered smoothly from $v$ to $u$ so
that it is owned contiguously by exactly one of the two nodes (avoiding an
ownership gap at time $t$).
We say that a token $\kappa$ is \emph{alive} at time
$t > 0$
if it is owned by some node at that time.

\subparagraph*{Restarts.}
\LongVersion 
Recall that a node
$v \in V$
invokes \ProcRestart{} if either
(1)
$v$ suffers a premature discovery;
(2)
the variable $v.\varTimer$ expires;
or
(3)
$v$ is a root signaled to invoke \ModuleTraversal{} while not holding a cold
token.
\LongVersionEnd 
When a physical node $v$ and its shadow node $\tilde{v}$ experience a restart,
a new token $\kappa$ is generated at $v$ (and immediately dispatched to
$\tilde{v}$).
We emphasize that if $v$ or $\tilde{v}$ hold a token $\kappa'$ in round $t$
prior to the invocation of \ProcRestart{}, then we think of $\kappa'$ as
disappearing and regard $\kappa$ as a new token;
this is another case where we say that $\kappa'$ \emph{dies} in round $t$.

Following the restart experienced by $v$ and $\tilde{v}$ in round $t$, the two
nodes form a tree of size $2$.
Moreover, after the restart, $v$ initiates a \ModuleProposing{} phase that
immediately invokes \ModuleFindCrossingEdge{}, lasting for
$\Cph \log N$
search epochs, during which $v$ does not \emph{interact} --- namely, send
messages or receive messages without ignoring them --- with any node
other than $\tilde{v}$ (recall that $\tilde{v}$ anyways cannot interact with
any node other than $v$).
As each search epoch lasts for
$2 \Ctr N$
rounds, we get the following observation.

\begin{observation} \label{observation:silence-after-restart}
Consider a physical (resp., shadow) node
$v \in V$
and suppose that $v$ experiences a restart in round
$t > 0$.
Then, $v$ does not interact with any node other than its corresponding shadow
(resp., physical) node at least until round
$t + \Cph \log N \cdot 2 \Ctr N + 1$.
\end{observation}

\subparagraph*{Compatible Components.}
Edge
$e = \{ u, v \} \in E$
is said to be \emph{compatible} at time $t$ if
$u = v.\varParent^{t}$
and
$v \in u.\varChildren^{t}$
(or vice versa).
Let
$G_{c}^{t} = (V, E_{c}^{t})$
be the undirected graph defined by setting
$E_{c}^{t} \subseteq E$
to be the subset of edges compatible at time $t$.
The (maximal) connected components of $G_{c}^{t}$ are referred to as
\emph{compatible components}.
By definition, $G_{c}^{t}$ is an undirected pseudoforest and the compatible
components are undirected pseudotrees.

Edge
$e = \{ v, w \} \in E$
is said to be a \emph{dangling} edge of node $v$ at time $t$ if $v$ regards
it as a tree edge, i.e.,
$w \in v.\varTreeNeighbors^{t}$,
and yet
$e \notin E_{c}^{t}$.
Notice that $e$ may be a dangling edge of $v$ if $w$ does not regard it as a
tree edge or if $w$ does regard it as a tree edge, but there is an
inconsistency in the parent-child relations between $v$ and $w$.
We say that $e$ is \emph{dangling} without mentioning $v$ if $v$ is not
important or clear from the context.

\subparagraph*{$\delta$-Components and Selfishness.}
Fix some node
$v \in V$.
We make an extensive use of the variable
$v.\delta \in v.\varTreeNeighbors \cup \{ v \} \cup \{ \bot \}$
defined for the sake of the analysis by setting
\[
v.\delta^{t}
\, = \,
\begin{cases}
v, & \text{if $v$ owns a token at time $t$} \\
v.\varTokenDirection^{t}, &
\text{if $v$ does not own a token at time $t$ and
    $\{ v, v.\varTokenDirection^{t} \} \in E_{c}^{t}$} \\
\bot, & \text{otherwise}
\end{cases} \, .
\]
Let
$G_{\delta}^{t} = (V, E_{\delta}^{t})$
be the directed graph defined by setting
$E_{\delta}^{t} =
\left\{ (v, v.\delta) \mid v \in V \, \land \, v.\delta \neq \bot \right\}$.
The weakly connected components of $G_{\delta}^{t}$ are referred to as
\emph{$\delta$-components}.
By definition, $G_{\delta}^{t}$ is a pseudoforest and the $\delta$-components
are pseudotrees.
Moreover, all edges of the undirected version of $G_{\delta}^{t}$ that
are not self loops are also edges of $G_{c}^{t}$.

A $\delta$-component is said to be \emph{selfish} if it includes a self loop.
A non-selfish $\delta$-component may be cycle free or include a cycle that
involves two or more nodes.
By definition, node
$v \in V$
is incident to a self loop in $G_{\delta}^{t}$ if and only if it owns a token
at time $t$, which means that the nodes of a selfish $\delta$-component own exactly one token (all together), whereas the nodes of a non-selfish $\delta$-component do not own any token.

For a token $\kappa$ that is alive at time $t$, let $D_{\kappa}^{t}$ be the
selfish $\delta$-component that includes the node that owns $\kappa$ and let
$\overline{D}_{\kappa}^{t}$ be the undirected version of $D_{\kappa}^{t}$
excluding the self loop (notice that $\overline{D}_{\kappa}^{t}$ is always a
tree).
We say that a node
$v \in V$
is \emph{covered} by $\kappa$ at time $t$ if
$v \in \V(D_{\kappa}^{t})$;
$v$ is said to be \emph{uncovered} at time $t$ if it belongs to a non-selfish
$\delta$-component in $G_{\delta}^{t}$.
Assuming that $\kappa$ is alive at time $t$, let $\kappa.\MaxTimer^{t}$
(resp., $\kappa.\MinTimer^{t}$) be the maximum (resp., minimum) value of the variables $v.\varTimer^{t}$ taken over all nodes
$v \in \V(D_{\kappa}^{t})$.

The selfish $\delta$-components obey the following dynamics.
Consider a token $\kappa$ that is dispatched by a node
$v \in V$
in round
$t - 1 > 0$
over edge
$e = \{ v, w \}$
and assume that the dispatch is successful so that $\kappa$ is acquired by node $w$ in round $t$.
If $w$ does not experience a restart in round $t$, then the $\delta$-component $D_{\kappa}$ is updated so that
(1)
its sole self loop moves from
$(v, v) \in E_{\delta}^{t}$
to
$(w, w) \in E_{\delta}^{t + 1}$;
and
(2)
the direction of $e$ changes from
$(w, v) \in E_{\delta}^{t}$
to
$(v, w) \in E_{\delta}^{t + 1}$.

\LongVersion 
Recalling that node
$v \in V$
does not own a token at time $t$ if and only if
$v.\varToken = 0$
and
$v.\varRecentPass = 0$,
we obtain the following observation by inspection of the conditions for
receiving the \msgPassToken{} and \msgRootTransfer{} messages as presented in
\Sect{}\ \ref{section:module-traversal} and
\ref{section:module-root-transfer}, respectively.
\LongVersionEnd 

\ShortVersion 
\addtocounter{theorem}{1}
\ShortVersionEnd 
\LongVersion 
\begin{observation} \label{observation:successful-dispatch}
Consider a node
$u \in V$
that dispatches a token $\kappa$ in round
$t - 1 > 0$
to an adjacent node
$v \in \Neighbors(u)$.
If $\kappa$ is hot (i.e., the dispatch is made over a \msgPassToken{}
message), then the dispatch succeeds if and only if
$v.\delta^{t} = u$.
If $\kappa$ is cold (i.e., the dispatch is made over a \msgRootTransfer{}
message), then the dispatch succeeds if and only if
$v.\delta^{t} = u$
and
$v.\varOutProposal \neq \bot$.
\end{observation}
\LongVersionEnd 

\LongVersion 
\Obs{}~\ref{observation:cannot-acquire-if-not-covered} now follows
directly from \Obs{}~\ref{observation:successful-dispatch}.
\LongVersionEnd 

\begin{observation} \label{observation:cannot-acquire-if-not-covered}
Consider a token $\kappa$ that is alive at time
$t > 0$
and a node
$v \in V$.
If
$v \notin \V(D_{\kappa}^{t})$,
then $v$ cannot acquire $\kappa$ in round $t$.
\end{observation}

\LongVersion 
Recalling that a node can modify its $\varParent$, $\varChildren$, and
$\varTokenDirection$ variables only when it holds a token or experiences a
restart, we obtain the following three observations.
\LongVersionEnd 

\ShortVersion 
\addtocounter{theorem}{1}
\ShortVersionEnd 
\LongVersion 
\begin{observation} \label{observation:remain-in-non-selfish-delta-component}
If a node
$v \in V$
is uncovered at time
$t > 0$
and it does not experience a restart in round $t$, then $v$ remains uncovered
at time
$t + 1$.
\end{observation}
\LongVersionEnd 

\begin{observation} \label{observation:remain-in-selfish-delta-component}
Consider a token $\kappa$ that is alive at time
$t > 0$
and is still alive at time
$t + 1$
and let $u$ be the node that owns $\kappa$ at time $t$.
Let $v$ be a node in $D_{\kappa}^{t}$ and let $P$ be the unique simple
$(v, u)$-path
in $D_{\kappa}^{t}$.
If none of the nodes in $P$ experiences a restart in round $t$, then
$v \in \V(D_{\kappa}^{t + 1})$.
\end{observation}

\ShortVersion 
\addtocounter{theorem}{1}
\ShortVersionEnd 
\LongVersion 
\begin{observation} \label{observation:leave-selfish-delta-component}
Consider a token $\kappa$ that is alive at time
$t > 0$
and is still alive at time
$t + 1$.
Let $v$ be a node that is covered by $\kappa$ at time $t$ and does not
experience a restart in round $t$.
If $v$ is not covered by $\kappa$ at time
$t + 1$,
then $v$ is uncovered at time
$t + 1$.
\end{observation}
\LongVersionEnd 

\subparagraph*{Strong Tokens.}
A token $\kappa$ is said to be \emph{strong} at time $t$ if
(1)
$\kappa$ is alive at time $t$;
(2)
the nodes of $D_{\kappa}^{t}$ admit no dangling edges at time $t$;
and
(3)
$\overline{D}_{\kappa}^{t}$ forms a compatible component in $G_{c}^{t}$.
\LongVersion 
\Obs{}~\ref{observation:successful-dispatch} guarantees
\LongVersionEnd 
\ShortVersion 
The algorithm is designed so
\ShortVersionEnd 
that the dispatch of a
strong hot token always succeeds, whereas the dispatch of a strong cold token
to a node
$v \in V$
succeeds if and only if
$v.\varOutProposal \neq \bot$.

If a token $\kappa$ is strong at time
$t > 0$,
then $\overline{D}_{\kappa}^{t}$ admits a (unique) root, i.e., a node $r$ with
$r.\varParent^{t} = \bot$,
and $v.\varParent^{t}$ points to the successor of $v$ along the unique simple
$(v, r)$-path
in $\overline{D}_{\kappa}^{t}$ for every node
$v \in \V(\overline{D}_{\kappa}^{t}) - \{ r \}$
(this is unrelated to the direction of the $v.\delta$ pointers in
$D_{\kappa}^{t}$ that always point towards the token).
Notice that this is not guaranteed if $\kappa$ is not strong as the edge
$\{ r, r.\varParent \}$
may be dangling for the node $r$ to which the $\varParent$ variables point.

\LongVersion 
\subparagraph*{The Visited Subgraph.}
Consider a token $\kappa$ that is alive throughout the time interval
$I = [t_{0}, t_{1}]$
for some
$0 < t_{0} < t_{1}$.
We say that $\kappa$ \emph{visits} node
$v \in V$
during $I$ if there exists some
$t \in I$
such that $\kappa$ is owned by $v$ at time $t$.
We say that $\kappa$ \emph{visits} edge
$e = \{ v, w \} \in E$
during $I$ if there exists some
$t \in I - \{ t_{1} \}$
such that $\kappa$ is acquired by $w$ in round $t$ after being dispatched by
$v$ over $e$ (in round
$t - 1$).
Notice that if $\kappa$ visits edge $e$ during $I$, then $\kappa$ visits both
endpoints of $e$ during $I$ which means that the nodes and edges visited by
$\kappa$ during $I$ form a (well defined) subgraph of $G$.
\LongVersionEnd 

\subparagraph*{Natural Traversals.}
Recall that a token $\kappa$ is hot whenever it is involved in a traversal of
\ModuleTraversal{}.
Suppose that $\kappa$ is hot at time
$t > 0$
and let $\rho$ be the traversal in which $\kappa$ is involved.
We say that $\rho$ is \emph{natural} if there exists some
$0 < t_{0} < t$
and
$r \in V$
such that
(1)
$r$ is a root at time $t_{0}$;
and
(2)
$r$ initiates $\rho$ in round $t_{0}$ (by turning $\kappa$ from cold to
hot).
In other words, a natural traversal $\rho$ is initiated by the algorithm,
rather than sprouting from the adversarially chosen initial configuration.
We say that a natural traversal $\rho$ of $\kappa$ initiated by $r$ in round
$t_{0}$ is \emph{accomplished} in round $t_{1}$ if $t_{1}$ is the first round
after $t_{0}$ such that $\kappa$ becomes cold at $r$ in round $t_{1}$;
notice that $\kappa$ may die before it becomes cold at $r$, in which
case $\rho$ is never accomplished.
The \emph{length} of a natural traversal that is initiated in round $t_{0}$ and
accomplished in round $t_{1}$ is defined to be
$t_{1} - t_{0} + 1$.
If $\kappa$ accomplishes a natural traversal in round $t_{1}$, then $\kappa$
is said to be \emph{mature} at all times
$t > t_{1}$
(at which it is alive).

\subsection{Fault Recovery}
\label{section:fault-recovery}
Our goal in this section is to establish the following theorem.

\begin{theorem} \label{theorem:fault-recovery}
Each node in $V$ may experience at most one restart throughout the execution.
Moreover, there exists some
$t^{*}_{r} = O (N)$
such that from time
$t^{*}_{r}$
onward,
(1)
all nodes are covered and do not experience any restart;
and
(2)
no token dies.
This means in particular that if a token $\kappa$ is alive at time
$t \geq t^{*}_{r}$
with
$\V(D_{\kappa}^{t}) = V$,
then $\kappa$ is alive at time
$t'$
with
$\V(D_{\kappa}^{t'}) = V$
for every
$t' > t$.
\end{theorem}

We refer to the time $t^{*}_{r}$ promised by
\Thm{}~\ref{theorem:fault-recovery} as the algorithm's \emph{fault recovery
time}.
The journey towards proving \Thm{}~\ref{theorem:fault-recovery} begins with
the following fundamental lemma.

\begin{lemma} \label{lemma:tree-shape}
Consider a token $\kappa$ that is alive throughout the time interval
$I = [t_{0}, t_{1}]$
for some
$0 < t_{0} < t_{1} \leq t_{0} + \Cph \log N \cdot 2 \Ctr N$.
The subgraph $T$ of $G$ formed by the nodes and edges visited by $\kappa$
during $I$ is a tree.
\end{lemma}
\LongVersion 
\begin{proof}
The subgraph $T$ is connected since $\kappa$ is always dispatched between
adjacent nodes.
To prove that $T$ is cycle free, we argue that for every node
$v \in V$
and incident edge
$e = \{ v, w \} \in E$,
if $\kappa$ is acquired by node $w$ over edge $e$ in round
$t_{0} \leq \check{t} < t_{1}$
and
$\check{t} < \hat{t} < t_{1}$
is the first round after $\check{t}$ during which $\kappa$ is acquired by $v$,
then the dispatch of $\kappa$ in round
$\hat{t} - 1$
must be preformed over $e$.

To this end, notice that
$v \in \V(D_{\kappa}^{\check{t} + 1})$
with
$v.\delta^{\check{t} + 1} = w$.
As long as $v$ remains in $D_{\kappa}$ without acquiring $\kappa$, the
variable $v.\delta$ remains set to
$v.\delta = w$,
hence \Obs{}~\ref{observation:successful-dispatch} ensures that $v$ cannot
acquire $\kappa$ through any edge other than $e$, while
\Obs{}~\ref{observation:cannot-acquire-if-not-covered} ensures
that $v$ cannot acquire any token other than $\kappa$.
\Obs{} \ref{observation:remain-in-non-selfish-delta-component} and
\ref{observation:leave-selfish-delta-component} imply that if $v$ leaves
$D_{\kappa}$ in some round
$t \geq \check{t} + 1$,
then it remains uncovered, and in particular cannot acquire $\kappa$ (by
\Obs{}~\ref{observation:cannot-acquire-if-not-covered}), until it
experiences a restart in some round
$t' \geq t$.
However, by \Obs{}~\ref{observation:silence-after-restart}, if $v$
experiences a restart in round $t'$, then it cannot acquire $\kappa$ at least
until round
$t' + \Cph \log N \cdot 2 \Ctr N + 1
>
t_{1}$.
\end{proof}
\LongVersionEnd 

Note that \Lem{}~\ref{lemma:tree-shape} actually holds without the
restriction that
$t_{1} \leq t_{0} + \Cph \log N \cdot 2 \Ctr N$,
but the arguments required to establish this stronger statement are developed
only later on.
\LongVersion 
\par
Recall that a token $\kappa$ is cold whenever it is held by a root
$r \in V$
outside the scope of an active traversal.
Specifically, $\kappa$ is cold at $r$ if either
(C1)
$r$ waits for the next epoch to begin after the traversal associated with the
current (search or accepting) epoch has terminated;
(C2)
$r$ holds the token as part of a root transfer epoch;
or
(C3)
$r$ has sent a \msgPropose{} message and is waiting for an \msgAccept{}
message as part of a proposing epoch.
\LongVersionEnd 
We can now establish the following two lemmas.

\begin{lemma} \label{lemma:bound-hot-interval}
A token $\kappa$ cannot remain (alive and) hot throughout the time interval
$I = [t, t + \Ctr N]$
for any
$t > 0$.
\end{lemma}
\LongVersion 
\begin{proof}
Assume towards contradiction that $\kappa$ remains hot throughout the
time interval $I$.
Let $T$ be the subgraph of $G$ formed by the nodes and edges visited by
$\kappa$ during $I$.
\Lem{}~\ref{lemma:tree-shape} guarantees that $T$ is a tree, hence it includes
at most
$n - 1 < N$
edges.
By the design of the algorithm, as long as $\kappa$ remains hot, it is
dispatched over every edge of $T$ at most once in each direction, which sums
up to less than
$2 N$
dispatches over the edges of $T$ during $I$.
A node
$v \in V$
that acquires $\kappa$ while it is hot may hold it for at most four
consecutive rounds (see \Sect{}~\ref{section:module-traversal}), hence each
dispatch of $\kappa$ over an edge of $T$ accounts for at most four rounds,
yielding less than
$8 N$
rounds in total.
A contradiction to the assumption that $\kappa$ remains hot throughout $I$ is
derived by taking
$\Ctr = 8$.
\end{proof}
\LongVersionEnd 

\begin{lemma} \label{lemma:bound-cold-interval}
A token $\kappa$ cannot remain (alive and) cold throughout the time interval
$I = [t, t + 6 \Ctr N]$
for any
$t > 0$.
\end{lemma}
\LongVersion 
\begin{proof}
The period of time in which $\kappa$ may be held by a root that waits for the
next epoch to begin (cause (C1)) lasts for less than
$2 \Ctr N$
rounds.
A root that sends a \msgPropose{} message may hold $\kappa$ (cause (C3))
for at most
$3 \Ctr N$
rounds.
Therefore, it remains to bound the number of rounds in which $\kappa$
may be held by nodes participating in a root transfer process (cause (C2)).

Consider a node
$v \in V$
that acquires $\kappa$ in round
$t' \in I$
as an intermediate node in a root transfer process.
The algorithm is designed so that
$v.\varOutProposal^{t'} \in v.\varChildren^{t'}$
and
$v.\varOutProposal^{t' + 2} = \bot$.
For this variable to get a value in $v.\varChildren$ again, $v$ must
participate in \ModuleFindCrossingEdge{} and in particular hold a hot token.
As $v$ belongs to $D_{\kappa}$ at time
$t' + 2$,
we deduce by \Obs{} \ref{observation:cannot-acquire-if-not-covered},
\ref{observation:remain-in-non-selfish-delta-component}, and
\ref{observation:leave-selfish-delta-component} that if $v$ owns a token other
than $\kappa$ at time
$t'' > t' + 2$,
then $v$ must have experienced a restart between time
$t' + 2$
and time $t''$.
However, by \Obs{}~\ref{observation:silence-after-restart}, this means
that $v$ cannot acquire $\kappa$ again during $I$.

Therefore, each node holds $\kappa$ as an intermediate node of a root transfer
process for at most two rounds during $I$ which sums up to at most
$2 N$
rounds in total.
This is guaranteed to be less than
$\Ctr N$
by requiring that
$\Ctr > 2$.
Combining causes (C1), (C2), and (C3), we conclude that a token cannot be
cold for
$6 \Ctr N$
consecutive rounds.
\end{proof}
\LongVersionEnd 

\ShortVersion 
\addtocounter{theorem}{2}
\ShortVersionEnd 
\LongVersion 
\Cor{} \ref{corollary:bound-natural-traversal-length} and
\ref{corollary:at-least-one-natural-traversal} follow since a token is hot if
and only if it is involved in a traversal.

\begin{corollary} \label{corollary:bound-natural-traversal-length}
The length of any natural traversal is smaller than
$\Ctr N$.
\end{corollary}

\begin{corollary} \label{corollary:at-least-one-natural-traversal}
If a token $\kappa$ is alive throughout the time interval
$[t_{0}, t_{0} + 8 \Ctr N]$,
$t_{0} > 0$,
then there exist some
$t_{0} \leq t < t' < t_{0} + 8 \Ctr N$
such that a natural traversal of $\kappa$ is initiated in round $t$ and
accomplished in round $t'$.
\end{corollary}
\LongVersionEnd 

The following key lemma demonstrates the important role played by
accomplishing natural traversals in establishing the fault recovery properties
of our algorithm.

\begin{lemma} \label{lemma:accomplish-natural-traversal}
If a token $\kappa$ accomplishes a natural traversal $\rho$ in round
$t - 1$,
then $\kappa$ is strong at time $t$ and all the nodes in $\V(D_{\kappa}^{t})$
have been visited by $\kappa$ during $\rho$.
\end{lemma}
\begin{proof}
Let $\rho$ be the natural traversal accomplished by $\kappa$ in round
$t - 1$
and let
$r \in V$
be the root in which $\rho$ starts in round $t_{0}$ and terminates in round
$t$.
As $\kappa$ is hot throughout $\rho$, \Lem{}~\ref{lemma:bound-hot-interval}
ensures that the length of $\rho$ is smaller than
$\Ctr N$.
Employing \Lem{}~\ref{lemma:tree-shape}, we conclude that the subgraph $T$
formed by the nodes and edges visited by $\kappa$ during $\rho$ is a tree.
We prove that
(1)
$T$ is a subgraph of $\overline{D}_{\kappa}^{t}$;
(2)
$T$ forms a compatible component in $G_{c}^{t}$;
and
(3)
the nodes in $\V(T)$ do not admit dangling edges at time $t$.
This establishes the assertion by recalling that $\overline{D}_{\kappa}^{t}$
is a subgraph of some compatible component in $G_{c}^{t}$.

Consider some node
$v \in \V(T)$
and let
$t_{0} \leq t^{d}_{v} < t$
be the discovery round of $v$ in $\rho$.
The fact that $\kappa$ is alive at time $t$ ensures that $v$ does not
experience a restart in round $t^{d}_{v}$.
Moreover, by \Obs{}~\ref{observation:silence-after-restart}, node $v$ cannot
experience a restart before round $t^{d}_{v}$ and still be discovered by
$\kappa$ in round $t^{d}_{v}$.

We prove by induction on the time
$t_{0} \leq t' \leq t$
that if
$t^{d}_{v} < t'$,
then
$v \in \V(D_{\kappa}^{t'})$
for every node
$v \in \V(T)$.
This holds vacuously for time
$t' = t_{0}$,
so assume that it holds for time
$t_{0} \leq t' < t$
and consider time
$t' + 1$.
Assume by contradiction that there exists some node $v$ with
$t^{d}_{v} < t' + 1$
that does not belong to
$D_{\kappa}^{t' + 1}$
and let $v$ be such a node that is closest to $\kappa$ in
$D_{\kappa}^{t'}$, noticing that
$t^{d}_{v} < t' + 1$
implies that
$v \in \V(D_{\kappa}^{t'})$
either because
$t^{d}_{v} < t'$,
in which case we can apply the inductive hypothesis, or because
$t^{d}_{v} = t'$,
which means that
$v \in \V(D_{\kappa}^{t'})$
by \Obs{}~\ref{observation:cannot-acquire-if-not-covered}.

Let $u$ be the node that owns $\kappa$ at time $t'$ and let $P$ be the unique
simple
$(v, u)$-path
in $D_{\kappa}^{t'}$.
\Obs{}~\ref{observation:remain-in-selfish-delta-component} implies that
some node in $\V(P)$ experiences a restart in round $t'$ and by the choice of
$v$ this node must be $v$.
This means that
$t^{d}_{v} < t'$
as otherwise, $\kappa$ dies in round
$t' < t$.
Moreover, the fact that
$v \in \V(D_{\kappa}^{t'})$
ensures that $v$ is not discovered by any token other than $\kappa$ in round
$t'$ (see \Obs{}~\ref{observation:cannot-acquire-if-not-covered}).
Therefore, $v$ must experience a restart in round $t'$ due to an expiration of
$v.\varTimer$.
A contradiction is now derived as
$t' < t < t_{0} + \Ctr N \leq t^{d}_{v} + \Ctr N$. 
It follows that $T'$ is a subgraph of $\V(D_{\kappa}^{t})$, thus establishing
(1), which means in particular that the nodes in $T$ do not experience a
restart during $\rho$.
Moreover, by definition, all edges in $\overline{D}_{\kappa}^{t}$, and thus
also in $T$, are compatible.
Since
$r.\varParent = \bot$
throughout $\rho$, it follows that if
$v \in \V(T) - \{ r \}$,
then $v.\varParent$ points to the successor of $v$ along the unique simple
$(v, r)$-path
in $T$ throughout the time interval
$(t^{d}_{v}, t]$.

Consider some node
$v \in \V(T)$
and recall that the algorithm is designed so that if
$w \in v.\varChildren$
when $v$ is retracted by $\rho$, then $\rho$ must have already
discovered and retracted from $w$.
Moreover, $v.\varChildren$ and $v.\varParent$ are not modified between $v$'s
retraction round and time $t$ as $v$ does not hold a token, nor does it
experience a restart, during this time interval.
Taking $H$ to be the compatible component of $r$ in $G_{c}^{t}$, we deduce by
induction on the distances from $r$ in $H$ that all nodes in $H$ are visited
by $\kappa$ during $\rho$, thus they are also nodes in $T$, establishing (2).

Finally, if $v$ admits a dangling edge
$e = \{ v, w \}$
at time $t$, then $e$ became dangling for $v$ after $v$ has already dispatched
$\kappa$ over $e$ as otherwise, $\kappa$ would have died during this
dispatch.
This means that $w$ has removed $v$ from $w.\varTreeNeighbors$ during
$\rho$.
However, it also means that
$w \in \V(T)$,
hence $w$ does not experience a restart during $\rho$.
This in turn implies that $w$ does not remove any node from
$w.\varTreeNeighbors$ during $\rho$, establishing (3).
\end{proof}

\Lem{}~\ref{lemma:accomplish-natural-traversal} guarantees that a token is
strong when it turns from hot to cold after a natural traversal;
\Lem{}~\ref{lemma:staying-cold-after-natural-traversal} guarantees that it
remains strong at least until it turns hot again.

\begin{lemma} \label{lemma:staying-cold-after-natural-traversal}
Consider a token $\kappa$ that accomplishes a natural traversal in round
$t - 1 > 0$
so that $\kappa$ is cold at time $t$.
Let
$t' \geq t$
be the latest time subsequent to $t$ at which $\kappa$ is still
(alive and) cold.
Then, $\kappa$ is strong at time $t'$ with
$\overline{D}_{\kappa}^{t'} = \overline{D}_{\kappa}^{t}$.
Moreover,
$\kappa.\MaxTimer^{t'} < 7 \Ctr N$,
$\kappa.\MinTimer^{t'} > \Ctr N$,
and
$\kappa$ does not die in round $t'$.
\end{lemma}

Notice that the claim that $\kappa$ does not die in round $t'$ means that
either
(1)
$\kappa$ turns hot in round $t'$ and initiates a new (natural) traversal;
or
(2)
$\kappa$ is dissolved in round $t'$ as part of a tree merger process.
\LongVersion 
\begin{proof}[Proof of \Lem{}~\ref{lemma:staying-cold-after-natural-traversal}]
Let $\rho$ be the traversal accomplished by $\kappa$ in round
$t - 1$
and let
$T = \overline{D}_{\kappa}^{t}$.
Applying \Lem{}~\ref{lemma:accomplish-natural-traversal} to $\rho$, we conclude
that $\kappa$ is strong at time $t$ and that all nodes in $\V(T)$ have been visited by $\kappa$ during $\rho$.
This means in particular that $T$ is a tree that forms a compatible component
in $G_{c}^{t}$ and $\kappa$ is owned at time $t$ by its (unique) root $r$.
Moreover, \Cor{}~\ref{corollary:bound-natural-traversal-length} ensures that
$v.\varTimer^{t} < \Ctr N$
for every node
$v \in \V(T)$.
As $\kappa$ is cold throughout the time interval
$I = [t, t']$,
\Lem{}~\ref{lemma:bound-cold-interval} guarantees that
$t' < t + 6 \Ctr N$,
hence
$v.\varTimer < 7 \Ctr N$
throughout $I$ for every node
$v \in \V(T)$.

By \Obs{}~\ref{observation:cannot-acquire-if-not-covered}, the nodes
in $\V(D_{\kappa})$ do not hold a hot token throughout $I$, thus they do not
suffer a premature discovery nor do they modify their $\varTreeNeighbors$
variables during $I$ unless their $\varTimer$ variables expire.
(The nodes that participate in \ModuleRootTransfer{} do modify their
$\varParent$ and $\varChildren$ variables, but they do so in a manner that
does not affect the union of these two variables, i.e., the set of their tree
neighbors.)
Since the $\varTimer$ variables of the nodes in
$\V(T) = \V(D_{\kappa}^{t})$
do not reach
$7 \Ctr N$
throughout $I$, it follows that these variables do not expire during $I$, hence no node in $T$ experiences a restart during $I$.
Therefore,
$\overline{D}_{\kappa} = T$
and
$\kappa$ is strong throughout $I$.
This also implies that
$\kappa.\MaxTimer^{t'} < 7 \Ctr N$.

Let $\eta$ be the (search or accepting) epoch with which $\rho$ is associated and let $t_{0}$ be the round in which $\eta$ (and $\rho$) are initiated.
Since $\kappa$ is held (cold) by $r$ from time $t$ until $\eta$ ends at time
$t_{0} + 2 \Ctr N$
and since
$t < t_{0} + \Ctr N$,
it follows that
$\kappa.\MinTimer^{t_{0} + 2 \Ctr N} > \Ctr N$.
We conclude that
$\kappa.\MinTimer^{t'} > \Ctr N$
as
$t' \geq t_{0} + 2 \Ctr N$.

It remains to prove that $\kappa$ does not die in round $t'$.
Since
$\kappa.\MaxTimer^{t'} < 7 \Ctr N$,
it follows that the $\varTimer$ variables of the nodes in $T$ do not expire in
round $t'$.
On the other hand,
$\kappa.\MinTimer^{t'} > \Ctr N$,
hence no node in $T$ suffers from a premature discovery in round $t'$.
Therefore, if $\kappa$ dies in round $t'$, then this is because it is
dispatched in round
$t' - 1$
from a node $v$ to an adjacent node $w$ and this dispatch fails (so that $w$
does not acquire $\kappa$) in round $t'$.
Since $\kappa$ is cold at time $t'$, it follows that the dispatch from $v$ to
$w$ is made over a \msgRootTransfer{} message.
This means that the epoch $\eta$ associated with $\rho$ is the safety epoch of
\ModuleFindCrossingEdge{} (see
\Sect{}~\ref{section:module-find-crossing-edge}) and that
$r.\varOutProposal^{t} \neq \bot$,
thus triggering the invocation of \ModuleRootTransfer{}.

The assumption that $\rho$ is natural ensures that the $\varOutProposal^{t}$
variables of the nodes in $T$ mark a simple path in $T$ from $r$ to a crossing
port $x$ and that $\kappa$ is dispatched, by means of \msgRootTransfer{}
messages, down this path, hence
$u.\varOutProposal \neq \bot$
for every node $u$ that receives a \msgRootTransfer{} message;
in particular,
$w.\varOutProposal^{t'} \neq \bot$.
As $\kappa$ is strong at time $t'$,
\Obs{}~\ref{observation:successful-dispatch} implies that the dispatch from
$v$ to $w$ does not fail in round $t'$.
\end{proof}
\LongVersionEnd 
\LongVersion 
The following two lemmas complement
\Lem{}~\ref{lemma:staying-cold-after-natural-traversal}
en route to establishing \Thm{}~\ref{theorem:fault-recovery}.
\LongVersionEnd 
\ShortVersion 
The following two lemmas, combined with
\Lem{}~\ref{lemma:staying-cold-after-natural-traversal}, allow us to establish
\Thm{}~\ref{theorem:fault-recovery};
refer to the attached full version for the full argument.
\ShortVersionEnd 

\begin{lemma} \label{lemma:late-search-traversals}
If a token $\kappa$ initiates a traversal $\rho$ associated with a search
epoch in round
$t^{i} > 8 \Ctr N$,
then $\rho$ is accomplished in some round
$t^{a} > t^{i}$
and
$\overline{D}_{\kappa}^{t^{a}} = \overline{D}_{\kappa}^{t^{i}}$.
\end{lemma}
\LongVersion 
\begin{proof}
Let $r$ be the root node at which $\rho$ is initiated.
If $\kappa$ is generated due to a restart during the time interval
$I = [t^{i} - 8 \Ctr N, t^{i}]$,
then the assertion is established by
\Obs{}~\ref{observation:silence-after-restart};
assume hereafter that $\kappa$ is alive throughout $I$.
\Cor{}~\ref{corollary:at-least-one-natural-traversal} ensures that $\kappa$
is mature at time $t^{i}$, so let $\rho_{0}$ be the last natural traversal
accomplished by $\kappa$ prior to time $t^{i}$, let
$t_{0}^{a} - 1$
be the round in which $\rho_{0}$ is accomplished, and let
$T = \overline{D}_{\kappa}^{t_{0}^{a}}$.
When applied to $\rho_{0}$,
\Lem{}~\ref{lemma:staying-cold-after-natural-traversal} guarantees that
$\kappa$ is strong at time $t^{i}$ with
$\overline{D}_{\kappa}^{t^{i}} = T$
whose unique root is $r$ and that
$\kappa.\MaxTimer^{t^{i}} < 7 \Ctr N$
and
$\kappa.\MinTimer^{t^{i}} > \Ctr N$.

We prove by induction on
$t \geq t^{i}$
that if $\kappa$ is still hot at time $t$, then $\kappa$ is strong at time $t$
with
$\overline{D}_{\kappa}^{t} = T$
whose unique root is $r$.
Assuming the inductive hypothesis at time $t$,
\Obs{}~\ref{observation:remain-in-non-selfish-delta-component} implies that it 
suffices to prove that the following requirements are satisfied for every node
$v \in \V(D_{\kappa}^{t})$:
(R1)
if $\kappa$ was dispatched to $v$ in round
$t - 1$,
then $v$ acquires $\kappa$ in round $t$;
(R2)
$v$ does not suffer from a premature discovery in round $t$;
(R3)
$v.\varTimer$ does not expire in round $t$;
(R4)
$v.\varParent$ is not modified in round $t$;
and
(R5)
$v$ does not obtain dangling edges in round $t$.

Requirement (R1) is satisfied by \Obs{}~\ref{observation:successful-dispatch}
as $\kappa$ is hot and strong at time $t$.
To see that requirement (R2) is satisfied, suppose that $v$ is discovered in
round $t$.
Since
$v \in \V(D_{\kappa}^{t})$,
it follows that this discovery is made by $\kappa$.
As the unique root in $\overline{D}_{\kappa}^{t}$ is $r$, we know that
$v.\varParent^{t} = v.\varParent^{t^{i}}$,
hence this is the first discovery of $v$ by $\kappa$ since time $t^{i}$.
Requirement (R2) is now established as
$v.\varTimer^{t} \geq \kappa.\MinTimer^{t^{i}} > \Ctr N$.
For requirement (R3), notice that \Lem{}~\ref{lemma:bound-hot-interval}
ensures that
$t < t^{i} + \Ctr N$
as $\kappa$ is still hot at time $t$, hence the bound
$\kappa.\MaxTimer^{t^{i}} < 7 \Ctr N$
implies that
$v.\varTimer^{t} < 8 \Ctr N$.
Combined with requirement (R2), we know that $v$ does not experience a
restart in round $t$.
This immediately yields requirement (R4) since $v.\varParent$ can be updated
only if $v$ experiences a reset or if $v$ owns a cold token.
To see that requirement (R5) is satisfied, notice that $v.\varTreeNeighbors$
can be updated in round $t$ only if $v$ processes a merger proposal and this
cannot happen as $\rho$ is associated with a search epoch.
\end{proof}
\LongVersionEnd 

\begin{lemma} \label{lemma:late-accepting-traversals}
If a token $\kappa$ initiates a traversal $\rho$ associated with an accepting
epoch in round
$t^{i} > 10 \Ctr N$,
then $\rho$ is accomplished in some round
$t^{a} > t^{i}$
and
$\overline{D}_{\kappa}^{t^{a}}$ is a supergraph of $\overline{D}_{\kappa}^{t^{i}}$.
\end{lemma}
\LongVersion 
\begin{proof}
Let $r$ be the root node at which $\rho$ is initiated and let $\eta$ be the
accepting epoch associated with $\rho$.
The token $\kappa$ cannot be generated due to a restart during the time
interval
$I = [t^{i} - 8 \Ctr N, t^{i}]$
as this would imply that $\eta$ is a search epoch by
\Obs{}~\ref{observation:silence-after-restart}, thus $\kappa$ is alive
throughout $I$.
\Cor{}~\ref{corollary:at-least-one-natural-traversal} ensures that $\kappa$
is mature at time $t^{i}$, so let $\rho_{0}$ be the last natural traversal
accomplished by $\kappa$ prior to time $t^{i}$, let
$t_{0}^{a} - 1$
be the round in which $\rho_{0}$ is accomplished, and let
$T' = \overline{D}_{\kappa}^{t_{0}^{a}}$.
When applied to $\rho_{0}$,
\Lem{}~\ref{lemma:staying-cold-after-natural-traversal} guarantees that
$\kappa$ is strong at time $t^{i}$ with
$\overline{D}_{\kappa}^{t^{i}} = T'$
whose unique root is $r$ and that
$\kappa.\MaxTimer^{t^{i}} < 7 \Ctr N$
and
$\kappa.\MinTimer^{t^{i}} > \Ctr N$.

We prove by induction on
$t \geq t^{i}$
that if $\kappa$ is still hot at time $t$, then $\kappa$ is strong at time $t$
with $\overline{D}_{\kappa}^{t}$ being a supergraph of $T'$
whose unique root is $r$.
Assuming the inductive hypothesis at time $t$,
\Obs{}~\ref{observation:remain-in-non-selfish-delta-component} implies that it 
suffices to prove that the following requirements are satisfied for every node
$v \in \V(D_{\kappa}^{t})$:
(R1)
if $\kappa$ was dispatched to $v$ in round
$t - 1$,
then $v$ acquires $\kappa$ in round $t$;
(R2)
$v$ does not suffer from a premature discovery in round $t$;
(R3)
$v.\varTimer$ does not expire in round $t$;
(R4)
$v.\varParent$ is not modified in round $t$;
and
(R5)
$v$ does not obtain dangling edges in round $t$.

If no incoming merger proposal is processed by the nodes in $\V(T')$
during $\rho$, then the requirements (R1)--(R5) are established by the same
line of arguments used in the proof of
\Lem{}~\ref{lemma:late-search-traversals}.
Suppose that node
$y \in \V(T')$
processes a merger proposal from node
$x \in \Neighbors(y)$
during $\rho$, sending an \msgAccept{} message to $x$ in round
$t' - 1$.
We argue that the following properties are satisfied:
(P1)
$x \notin \V(D_{\kappa}^{t'})$;
(P2)
$x$ owns a strong token $\kappa'$ at time $t'$; 
(P3)
$x$ is waiting for an \msgAccept{} message from $y$ at time $t'$ as part of a
proposing epoch;
(P4)
$\kappa'.\MaxTimer^{t'} < 7 \Ctr N$
and
$\kappa'.\MinTimer^{t'} > \Ctr N$;
and
(P5)
$v.\varAcceptingProposals^{t'} = 0$
for every node
$v \in \V(D_{\kappa'}^{t'})$.

Properties (P1)--(P5) suffices to complete the proof:
First, they ensure that
$T = \overline{D}_{\kappa'}^{t'}$
is merged into $\overline{D}_{\kappa}^{t'}$ in round $t'$, dissolving
$\kappa'$, and that $\kappa$ is strong at time
$t' + 1$
with
$\overline{D}_{\kappa}^{t' + 1}$
being the merged tree.
Second, $\kappa$ is dispatched from $y$ to $x$ in round $t'$ as part of
$\rho$, triggering a (sub)traversal $\rho^{T}$ of $T$.
Since the $\varAcceptingProposals$ variables of the nodes in $\V(T)$ are
turned on only during their retraction rounds in $\rho^{T}$, it follows that
they do not process any incoming merger proposals during $\rho^{T}$.
Therefore, we can repeat the line of arguments from the proof of
\Lem{}~\ref{lemma:late-search-traversals} to conclude that $\kappa$ remains
strong throughout $\rho^{T}$ and $\rho^{T}$ is accomplished with the
$\varParent$ variables of the nodes in $T$ pointing towards $y$.
The assertion is established by applying this argument to all merger
proposals processed during $\rho$.

It remains to prove that the aforementioned properties (P1)--(P5) are
satisfied.
To this end, let $\eta_{0}$ be the epoch associated with the natural traversal
$\rho_{0}$ of $T'$ accomplished by $\kappa$ in round
$t_{0}^{a} - 1$
and let $t_{0}^{i}$ be the round in which $\rho_{0}$ (and $\eta_{0}$) are
initiated.
We argue that $\eta_{0}$ cannot be a search epoch:
indeed, this would imply that $\eta$ (the epoch associated with $\rho$) is the
first accepting epoch of its \ModuleAccepting{} phase, with the previous phase
being a $\ModuleProposing{}$ phase, in which case,
$y.\varAcceptingProposals = 0$
from time $t^{i}$ until $y$'s retraction round in $\rho$;
in particular, we get that
$y.\varAcceptingProposals^{t' - 1} = 0$,
contradicting the assumption that $y$ sends an \msgAccept{} message in round
$t' - 1$.
Therefore, $\eta_{0}$ is also an accepting epoch which means that
$t^{i} = t_{0}^{i} + 2 \Ctr N$.

The algorithm is designed so that
$y.\varInProposals(x) = 0$
when $\rho_{0}$ retracts from $y$ in round
$t_{0}^{i}
<
t_{0}^{r} - 1
\leq
t_{0}^{a}$
(see \Sect{}~\ref{section:module-accepting}).
Since
$y.\varInProposals^{t' - 1}(x) = 1$,
it follows that $y$ receives a \msgPropose{} message from $x$ and sets
$y.\varInProposals(x) \gets 1$
in some round
$t_{0}^{r} \leq t^{p} + 1 < t' - 1$,
which means that $x$ sends this message and starts a proposing epoch
$\eta^{p}$ in round
$t^{p}
\geq
t_{0}^{r} - 1
>
t_{0}^{i}
=
t^{i} - 2 \Ctr N$;
let $\kappa'$ be the cold token held by $x$ in round $t^{p}$.
On the other hand, as $\kappa$ turns from cold to hot in round $t^{i}$ and it
is still hot at time
$t' + 1$,
\Lem{}~\ref{lemma:bound-hot-interval} ensures that
$t'
<
t^{i} + \Ctr N
<
t^{p} + 3 \Ctr N$,
thus $\eta^{p}$ does not end before $x$ receives the \msgAccept{} message from
$y$ in round $t'$, establishing property (P3).
In particular, $x$ still owns $\kappa'$ at time $t'$.

Next, as
$t^{p} > t^{i} - 2 \Ctr N > 8 \Ctr N$,
\Cor{}~\ref{corollary:at-least-one-natural-traversal} ensures that
$\kappa'$ is mature at time
$t^{p}$.
Recalling that $\kappa'$ remains cold between time $t^{p}$ and time $t'$, we
can apply \Lem{}~\ref{lemma:staying-cold-after-natural-traversal} to the last
(natural) traversal $\rho'$ accomplished by $\kappa'$ prior to time
$t^{p}$ to conclude that $\kappa'$ is strong at time $t'$ and that
$\kappa'.\MaxTimer^{t'} < 7 \Ctr N$
and
$\kappa'.\MinTimer^{t'} > \Ctr N$,
establishing properties (P2) and (P4).
By definition, $\rho'$ must be associated with a search epoch, hence
$v.\varAcceptingProposals^{t'} = v.\varAcceptingProposals^{t^{p}} = 0$
for every node
$v \in \V(D_{\kappa'}^{t'})$,
establishing property (P5).
Finally, as
$y.\varAcceptingProposals^{t'} = 1$,
we conclude that
$y \notin \V(D_{\kappa'}^{t'})$,
thus
$x \notin \V(D_{\kappa}^{t'})$,
establishing property (P1).
\end{proof}
\LongVersionEnd 

\LongVersion 
The proof of \Thm{}~\ref{theorem:fault-recovery} can now be completed:
Combining \Obs{}~\ref{observation:silence-after-restart} with
\Cor{}~\ref{corollary:at-least-one-natural-traversal}, we conclude that if a
token $\kappa$ is alive at time
$18 \Ctr N$,
then there exists some
$10 \Ctr N < t \leq 18 \Ctr N$
such that $\kappa$ initiates a traversal in round $t$.
\Lem{} \ref{lemma:staying-cold-after-natural-traversal},
\ref{lemma:late-search-traversals}, and
\ref{lemma:late-accepting-traversals} ensure that $\kappa$ does not die from
round $t$ onward.
These three lemmas also guarantee that if node
$v \in V$
is covered at time
$t \geq 18 \Ctr N$,
then it remains covered indefinitely and does not experience any restart after
time $t$.
On the other hand, \Obs{}
\ref{observation:cannot-acquire-if-not-covered} and
\ref{observation:remain-in-non-selfish-delta-component}
imply that if $v$ is uncovered at time
$18 \Ctr N$,
then it must experience a restart by round
$18 \Ctr N + 8 \Ctr N = 26 \Ctr N$,
thus becoming covered.
Finally, by \Obs{}~\ref{observation:silence-after-restart}, node $v$ may
experience at most one restart up to time
$18 \Ctr N$,
in which case it is covered at time
$18 \Ctr N$
and hence, does not experience any more restarts.
\Thm{}~\ref{theorem:fault-recovery} follows by setting
$t^{*}_{r} = 26 \Ctr N + 1$.
\LongVersionEnd 

\subsection{Stabilization Time}
\label{section:stabilization-time}
Recalling that $t^{*}_{r}$ is the fault recovery time promised in
\Thm{}~\ref{theorem:fault-recovery}, our goal in this section is to prove the
following theorem.

\begin{theorem} \label{theorem:stabilization-time}
Let
$t^{*}_{s} \geq t^{*}_{r}$
be the earliest time such that exactly one token is alive at time
$t^{*}_{s}$.
Then
$t^{*}_{s} = O (N \log^{2} N)$
in expectation and whp.
\end{theorem}

\Thm{}~\ref{theorem:stabilization-time} establishes the stabilization
properties of our algorithm:
Recall that after time $t^{*}_{r}$, a single token is alive if and only if
this token covers the whole graph.
Therefore, from time $t^{*}_{s}$ onward, the graph admits no crossing edges
and any invocation of \ModuleProposing{} excludes the root transfer and
proposing epochs.
This means that the $\varParent$ and $\varChildren$ variables of the nodes in
$V$ remain unchanged and, in particular, the (single) root remains fixed.

Let
$\Lph = \Cph \log N \cdot 2 \Ctr N$
and
$\Lph^{+} = \Lph + 4 \Ctr N$,
recalling that each \ModuleAccepting{} phase lasts exactly $\Lph^{+}$ rounds
and each \ModuleProposing{} phase lasts between $\Lph$ and $\Lph^{+}$ rounds.

\begin{lemma} \label{lemma:token-dissolves-with-constant-probability}
Fix some
$t_{0} \geq t^{*}_{r}$
and the configuration at time $t_{0}$ and consider a token $\kappa$ that is
alive at time $t_{0}$.
There exists a universal constant
$p > 0$
such that with probability at least $p$, independently of any coin tossed
outside the time interval
$[t_{0}, t_{0} + 4 \Lph)$,
either
(1)
$\kappa$ is no longer alive at time
$t_{0} + 4 \Lph$;
or
(2)
$\V(D_{\kappa}^{t_{0} + 4 \Lph}) = V$.
\end{lemma}
\begin{proof}
If $\kappa$ does not start a phase during the time interval
$I = [t_{0} + 5 \Ctr N, t_{0} + 5 \Ctr N + \Lph^{+})$,
then $\kappa$ must have been dissolved by time
$t_{0} + 5 \Ctr N + \Lph^{+} < t_{0} + 4 \Lph$,
thus completing the proof.
Assume that there exists some
$t \in I$
such that $\kappa$ starts a phase in round $t$;
let $t$ be the earliest such round and let $\phi$ be the corresponding phase.
If
$\V(D_{\kappa}^{t}) = V$,
then
$\V(D_{\kappa}^{t_{0} + 4 \Lph}) = V$,
thus completing the proof, so assume that
$\V(D_{\kappa}^{t}) \subset V$.

With probability
$1 / 2$,
the phase $\phi$ is a \ModuleProposing{} phase;
condition hereafter on this event.
Since
$\V(D_{\kappa}^{t}) \subset V$,
it follows that \ModuleFindCrossingEdge{} returns an edge
$e = \{ x, y \}$
that crosses between
$\V(D_{\kappa}^{t}) \ni x$
and
$V - \V(D_{\kappa}^{t}) \ni y$
with probability low-bounded by a positive constant;
condition hereafter on this event as well.

Let
$t + \Lph
\leq
t_{p} - 1
<
t + \Lph + \Ctr N$
be the round in which $x$ sends a \msgPropose{} message to $y$ as part of the
proposing epoch of $\phi$.
The algorithm is designed so that if
$y.\varAcceptingProposals^{\hat{t}} = 1$
for every
$t_{p}
\leq
\hat{t}
<
t_{p} + 3 \Ctr N$,
then $x$ is dissolved by round
$t_{p} + 3 \Ctr N
\leq
t + \Lph + 4 \Ctr N
<
t_{0} + 4 \Lph$
due to an \msgAccept{} message received from $y$, thus
establishing the assertion.
The proof is completed by showing that this event occurs with probability at
least
$1 / 4$.

\Thm{}~\ref{theorem:fault-recovery} ensures that node $y$ is covered by a
token $\kappa'$ at time $t_{p}$.
Moreover, since
$t_{p}
>
t + \Lph
\geq
t_{0} + \Lph + 5 \Ctr N
=
t_{0} + \Lph^{+} + \Ctr N$,
it follows that there exists some
$t_{0}
<
t'
\leq
t_{p} - \Ctr N$
such that $\kappa'$ starts a phase in round $t'$;
let $t'$ be the latest such round and let $\phi'_{1}$ be the corresponding
phase, noticing that
$t_{p} - \Lph^{+} - \Ctr N
<
t'
\leq
t_{p} - \Ctr N$.

The key observation now is that since the length of an \ModuleAccepting{}
phase is exactly $\Lph^{+}$ and the length of a \ModuleProposing{} phase is at
most $\Lph^{+}$, it follows that the probability that $\phi'_{1}$ is an
\ModuleAccepting{} phase, given the assumption that it starts in round $t'$
and does not end before round 
$t_{p} - \Ctr N$,
is at least
$1 / 2$;
condition hereafter on this event, recalling that
$t' > t_{0}$.
Given that $\phi'_{1}$ is an \ModuleAccepting{} phase, the token $\kappa'$ is
guaranteed to start another phase $\phi'_{2}$ in round
$t' + \Lph^{+}$
and with probability
$1 / 2$,
this phase is also an \ModuleAccepting{} phase;
condition hereafter on this event as well, observing that
$t' + \Lph^{+}
<
t_{p} + \Lph^{+} - \Ctr N
\leq
t + \Lph + \Lph^{+}
<
t_{0} + \Lph + 2 \Lph^{+} + 5 \Ctr N
<
t_{0} + 4 \Lph$.

As $\phi'_{1}$ is accepting, we know that by the end of the traversal
associated with its first (accepting) epoch, the variable
$y.\varAcceptingProposals$ must be set to $1$;
employing
\LongVersion 
\Cor{}~\ref{corollary:bound-natural-traversal-length},
\LongVersionEnd 
\ShortVersion 
\Lem{}~\ref{lemma:bound-hot-interval} and the fact that the token is hot
throughout this traversal,
\ShortVersionEnd 
we conclude that
$y.\varAcceptingProposals^{t' + \Ctr N} = 1$.
Since $\phi'_{2}$ is also accepting, it follows that
$y.\varAcceptingProposals$ remains set to $1$ at least until the end of
$\phi'_{2}$ at time
$t' + 2 \Lph^{+}$.
The assertion is now established as
$t_{p} \geq t' + \Ctr N$
and
$t_{p} + 3 \Ctr N
<
t' + \Lph^{+} + 4 \Ctr N
<
t' + 2 \Lph^{+}$.
\end{proof}

For
$i \in \Integers_{\geq 0}$,
let $K_{i}$ be the set of tokens that are alive at time
$t^{*}_{r} + 4 i \Lph$
and let
$K = K_{0}$.
Let
$Z_{\kappa} =
\min \{ i \in \Integers_{\geq 0} : \kappa \notin K_{i} \lor |K_{i}| = 1 \}$
for every token
$\kappa \in K$
and let
$Z = \max \{ Z_{\kappa} : \kappa \in K \}$.
\Lem{}~\ref{lemma:token-dissolves-with-constant-probability} guarantees that
$\Pr(Z_{\kappa} \geq i + 1 \mid Z_{\kappa} \geq i) \leq 1 - p$,
hence
$Z_{\kappa} \leq O (\log n)$
whp for every
$\kappa \in K$.
By the union bound, we deduce that
$Z \leq O (\log n)$
whp.
Since the same argument can be applied with
$K = K_{j \cdot h}$
for
$h = \Theta (\log n)$
and every
$j \in \Integers_{\geq 0}$,
it follows that
$Z \leq O (\log n)$
also in expectation, thus establishing \Thm{}~\ref{theorem:stabilization-time}.

\subsection{Bounding the Number of Messages}
\label{section:message-complexity}
We now turn to bound the message complexity%
\LongVersion 
, and through it, the bit complexity,
\LongVersionEnd 
\ShortVersion 
\
\ShortVersionEnd 
of our algorithm.

\begin{theorem} \label{theorem:message-complexity}
Until it stabilizes, the algorithm sends
$O (N \log^{2} N)$
messages in expectation and whp.
After stabilizing, the algorithm sends at most one message per round.
\end{theorem}
\begin{proof}
Node
$v \in V$
sends at most one message per round, thus
$O (N)$
messages are sent in round $0$.
In rounds
$t > 0$,
node $v$ may send a message only if it holds a token, thus establishing the
desired bound
\LongVersion 
on the message complexity
\LongVersionEnd 
after stabilization, when one token is alive.
It remains to bound the number of messages sent during the time interval
$[1, t^{*}_{s})$,
where $t^{*}_{s}$ is the earliest time so that one token is alive at
time $t^{*}_{s}$.
This is done separately for each message type.

Consider a \msgPassToken{} message $M$ sent from node
$v \in V$
to a node
$u \in \Neighbors(v)$.
If $u$ ignores $M$, then the token that $M$ carries dies.
\Thm{}~\ref{theorem:fault-recovery} ensures that the total number of distinct
tokens that existed throughout the execution is up-bounded by
$2 N$,
hence
$O (N)$
\msgPassToken{} messages are ignored throughout the execution%
\LongVersion 
;
we focus hereafter on the \msgPassToken{} messages that are not ignored.
\LongVersionEnd 
\ShortVersion 
.
\ShortVersionEnd 

Consider a node
$v \in V$.
Recalling that $v$ suffers a premature discovery (and invokes \ProcRestart{})
if it receives an (unignored) discovery \msgPassToken{} message while
$v.\varTimer < \Ctr N$,
we conclude that $v$ receives
$O (1)$
discovery \msgPassToken{} messages throughout any time interval of length
$\Ctr N$.
Moreover, $v$ must receive a discovery \msgPassToken{} message between any two
retraction \msgPassToken{} messages it sends, thus $v$ sends
$O (1)$
retraction \msgPassToken{} messages throughout any time interval of length
$\Ctr N$.

Since any \msgPassToken{} message
\LongVersion 
sent over edge
$e \in E$
\LongVersionEnd 
is either a discovery message or a retraction message%
\LongVersion 
{} of one of $e$'s endpoints%
\LongVersionEnd 
, it follows that we can charge all the (unignored) \msgPassToken{}
messages sent throughout the execution to the nodes so that each node
$v \in V$
is charged with
$O (1)$
messages in the time interval
$[1 + i \Ctr N, 1 + (i + 1) \Ctr N)$
for any
$i \in \Integers_{\geq 0}$.
Therefore, the total number of \msgPassToken{} messages sent during the time
interval
$[1, t^{*}_{s})$
is
$O (N) + O (N) \cdot \frac{t^{*}_{s}}{N}
=
O (N + t^{*}_{s})$.

An upper bound of
$O (N + t^{*}_{s})$
on the number of \msgRootTransfer{} (resp., \msgPropose{}) messages sent
during
$[1, t^{*}_{s})$
follows from the observation that a node must receive (and send) at least one
\msgPassToken{} message between any two \msgRootTransfer{} (resp., \msgPropose{}) messages it sends.
Finally, each \msgAccept{} message can be charged to either a \msgPropose{}
message or to a dying token;
the latter accounts to
$O (N)$
additional messages.
\LongVersion 
\par
\LongVersionEnd 
To sum up, the total number of messages sent during the time interval
$[1, t^{*}_{s})$
is
$O (N + t^{*}_{s})$.
\Thm{}~\ref{theorem:message-complexity} follows from
\Thm{}~\ref{theorem:stabilization-time} ensuring that
$t^{*}_{s} \leq O (N \log^{2} N)$
in expectation and whp.
\end{proof}

\section*{Acknowledgment}
We are grateful to Valerie King for helpful discussions.

\bibliographystyle{alpha}
\bibliography{references}

\clearpage
\appendix

\begin{figure}[!t]
\centering{\LARGE{APPENDIX}}
\end{figure}

\begin{table}[h]
\begin{tabular}{|l|l|l|}
\hline
\textbf{Variable} & \textbf{Range} & \textbf{Semantics} \\
\hline
$\varParent$ &
$\Neighbors(v) \cup \{ \bot \}$ &
$v$'s parent \\
\hline
$\varChildren$ &
$(\Neighbors(v))^{*}$ &
$v$'s children \\
\hline
$\varTreeNeighbors$ &
$2^{\Neighbors(v)}$ &
$v$'s tree neighbors \\
\hline
$\varToken$ &
$\{ 0, 1 \}$ &
indicates that $v$ holds a token \\
\hline
$\varTokenDirection$ &
$\varTreeNeighbors$ &
the direction of the token \\
\hline
$\varRecentPass$ &
$\{ 0, 1 \}$ &
indicates that $v$ has passed a token in the previous round \\
\hline
$\varTimer$ &
$\Integers_{\geq 0}$ &
\#rounds since receiving a (hot) token \\
\hline
$\varOutProposal$ &
$\Neighbors(v) \cup \{ \bot \}$ &
the direction of the (recently found) crossing edge \\
\hline
$\varInProposals(u)$ &
$\{ 0, 1 \}$ &
indicates that a proposal from
$u \in \Neighbors(v)$
has been registered \\
\hline
\end{tabular}
\caption{\label{table:variables}
The variables of node
$v \in V$.}
\end{table}

\begin{figure}
{\centering
\includegraphics[width=0.9\textwidth]{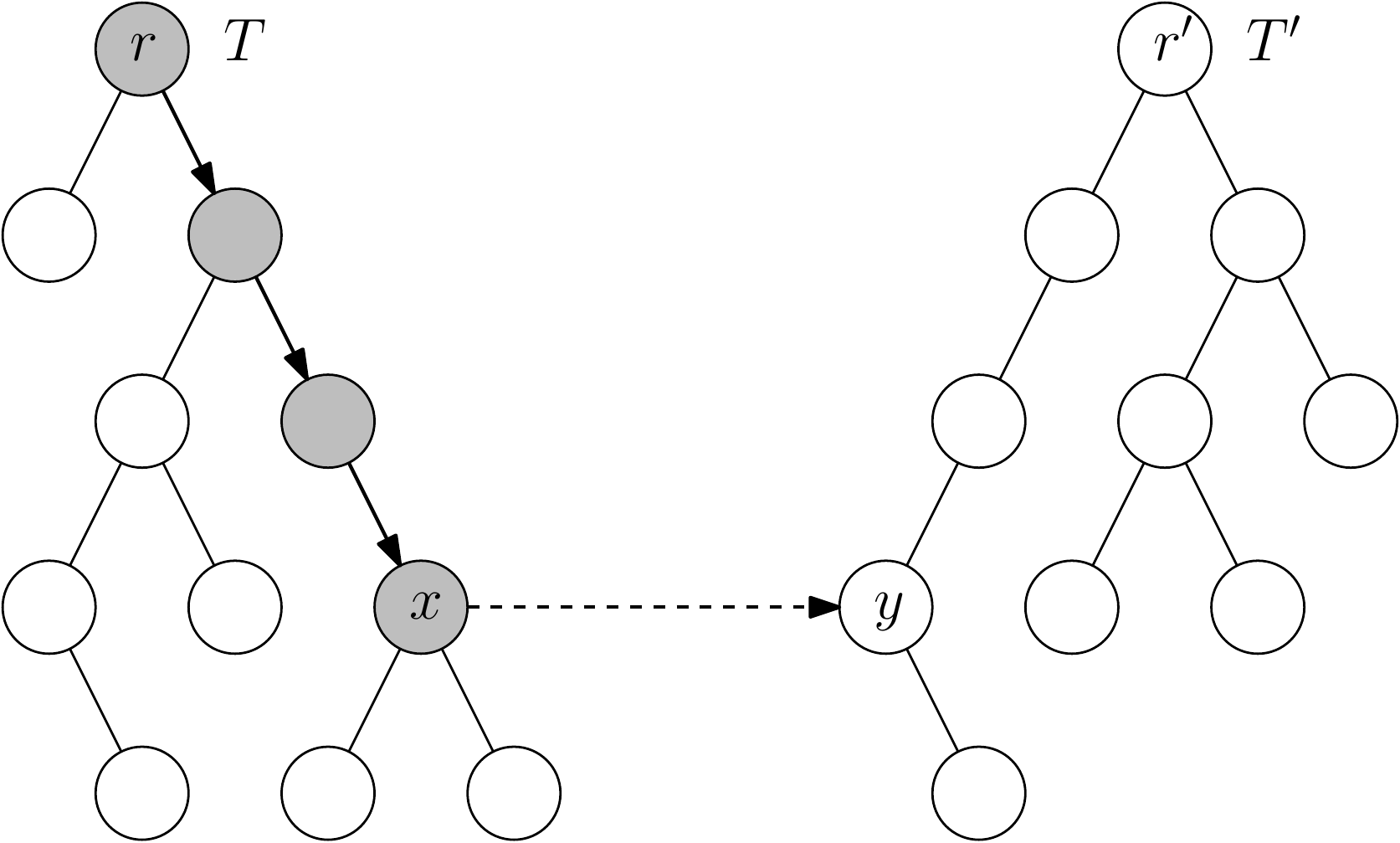}
\par}
\caption{\label{figure:tree-merge}
A merger of the proposing tree $T$ rooted at $r$ and the accepting tree $T'$
rooted at $r'$ over the crossing edge
$\{ x, y \}$.
The tree edges are depicted by the solid lines whereas the crossing edge is
depicted by the dashed line.
The nodes along the unique
$(r, x)$-path
in $T$ are marked in gray and the values of their $\varOutProposal$ variables
are depicted by the oriented edges.
These variables are set by \ModuleFindCrossingEdge{} during the last (safety)
search epoch.
The root transfer epoch is then dedicated to transferring the root of $T$,
together with the cold token, from $r$ to $x$.
The tree merger is triggered at the beginning of the proposing epoch by a
\msgPropose{} message sent from $x$ to $y$.
The merger is executed and $x$ becomes a child of $y$ in the merged tree if
$x$ receives an \msgAccept{} message from $y$ in the subsequent
$3 \Ctr N - 1$
rounds;
otherwise, this merger attempt fails and a new phase starts at $T$ that is
now rooted at $x$.
From the perspective of $y$, the proposal of $x$ is processed, sending back an
\msgAccept{} message, when $y$ is retraction ready with respect to a traversal
associated with one of the accepting epochs orchestrated by $r'$.
After the merger is completed, $y$ passes its token to $x$ as part of the
traversal process.
}
\end{figure}

\end{document}